\documentclass{article}

\usepackage{amsmath,amssymb, amsthm, mathrsfs,verbatim} 
\usepackage{cite,fullpage}
\usepackage{amsmath,amssymb,amsfonts}
\usepackage{algorithmic}
\usepackage{graphicx}
\usepackage{textcomp}
\usepackage{xcolor}
\def\BibTeX{{\rm B\kern-.05em{\sc i\kern-.025em b}\kern-.08em
    T\kern-.1667em\lower.7ex\hbox{E}\kern-.125emX}}

\usepackage{cite}
\usepackage{amsmath}
\usepackage{amssymb,amsfonts}
\usepackage{graphicx}
\usepackage{textcomp}
\usepackage{xcolor}
\def\BibTeX{{\rm B\kern-.05em{\sc i\kern-.025em b}\kern-.08em
    T\kern-.1667em\lower.7ex\hbox{E}\kern-.125emX}}
\usepackage{amsmath,amssymb,mathrsfs,verbatim} 
\usepackage{xspace}
\usepackage{bbold}
\usepackage{xcolor,setspace}
\usepackage[linesnumbered,lined,boxed]{algorithm2e}
\usepackage{fullpage,times,enumerate}
\usepackage{algorithmic}
\definecolor{verde}{rgb}{0.,0.6,0.2}
\definecolor{bianco}{rgb}{1.,1.,1.}
\definecolor{marrone}{rgb}{0.7,0.2,0.1}
\definecolor{rosso}{rgb}{1,0,0}
\definecolor{giallo}{rgb}{1.0, 0.87, 0.0}
\definecolor{blu}{rgb}{0.03, 0.27, 0.49}
\definecolor{daffodil}{rgb}{0.03, 0.27, 0.49}
\definecolor{darkcerulean}{rgb}{1.0, 0.87, 0.0}
\newcommand{\greenblue}[1]{{\color{verde}{#1}}}
\newcommand{\blue}[1]{{\color{daffodil}{#1}}}

\newcommand{\red}[1]{{\color{rosso}{#1}}}
\newcommand{\yellow}[1]{{\color{darkcerulean}{#1}}}

\newcommand{\pp}{\ensuremath{\mathbb {P}}\xspace}
\newcommand{\co}{\ensuremath{\mathbb {c}}\xspace}

\DeclareMathOperator{\fail}{\mathsf{fail}}

\DeclareMathOperator{\ID}{\mathsf{ID}}
\DeclareMathOperator{\BNT}{\mathsf{BNT}}
\DeclareMathOperator{\SEP}{\mathsf{SEP}}
\DeclareMathOperator{\DIS}{\mathsf{DIS}}
\DeclareMathOperator{\MHS}{\mathsf{MHS}}
\DeclareMathOperator{\MNS}{\mathsf{MNS}}

\DeclareMathOperator{\Bin}{\mathsf{Bin}}

\DeclareMathOperator{\BAD}{\mathsf{BAD}}
\DeclareMathOperator{\NP}{\mathsf{NP}}
\DeclareMathOperator{\GOOD}{\mathsf{GOOD}}

\usepackage{graphicx}

\usepackage{float}
\usepackage{subcaption}

\usepackage{tikz}
\usepackage{qtree}
\usepackage{multirow}
\usetikzlibrary{matrix}
\usetikzlibrary{positioning}

\usetikzlibrary{matrix}
\usetikzlibrary{positioning}
\usetikzlibrary{shapes.geometric}
\usetikzlibrary{shapes,decorations,shadows}  
\usetikzlibrary{decorations.pathmorphing}  
\usetikzlibrary{decorations.shapes}  
\usetikzlibrary{fadings}  
\usetikzlibrary{patterns}  
\usetikzlibrary{calc}  
\usetikzlibrary{automata}
\usetikzlibrary{chains,fit,shapes}
\usetikzlibrary{backgrounds}

\newtheorem{theorem}{Theorem}[section]
\newtheorem{definition}[theorem]{Definition}

\newtheorem{lemma}[theorem]{Lemma}
\newtheorem{corollary}[theorem]{Corollary}

\usetikzlibrary{decorations}
\tikzset{
bicolor/.style 2 args={
  dashed,dash pattern=on 20pt off 30pt,-,#1,
  postaction={draw,dashed,dash pattern=on 20pt off 30pt,-,#2,dash phase=20pt}
  },
}

\usepackage{url}    
    
\begin{document}

\title{Counting and localizing defective  nodes by Boolean network tomography}
\author{Nicola Galesi \\ Sapienza Universit\`a Roma  \\ \textit{Department of Computer Science}\\
Rome, Italy   \and  Fariba Ranjbar \\ Sapienza Universit\`a Roma  \\
\textit{Department of Computer Science} \\
Rome, Italy \\}

\maketitle

\begin{abstract}
Identifying defective items in larger sets is a main problem with many applications in real life situations.  We consider the problem of 
localizing defective nodes in networks through an approach based on boolean network tomography ($\BNT$), which is grounded on inferring informations 
from the boolean outcomes of end-to-end measurements paths.  {\em Identifiability} conditions on  the set of paths which guarantee  
discovering or counting unambiguously the defective nodes are of course very relevant. 
We investigate old and introduce new identifiability conditions contributing this problem both from a theoretical and applied perspective. 
%(1) What are the limits on the number of nodes/paths to guarantee unambiguous  identifiability      of $k$ defective nodes? 
  (1) What is the precise tradeoff between number of nodes and number of paths such that at most $k$ nodes can be identified unambiguously ?   The answer is known only for $k=1$ and we answer the question for any $k$, setting a problem implicitly left open in previous works. (2) We study upper and lower bounds on the number of unambiguously identifiable  nodes, introducing new identifiability conditions which strictly imply and are strictly implied by unambiguous identifiability;  (3) We use these new conditions on one side to design algorithmic  heuristics to count defective nodes in a fine-grained way, on the other side to prove the first complexity hardness results on the problem of identifying defective nodes in networks via $\BNT$. (4) We introduce a random model where we study lower bounds on the number of unambiguously identifiable defective nodes and we use  this model  to estimate that number on real networks by  a maximum likelihood estimate approach.\end{abstract}

%\begin{keywords}
%failure nodes, boolean network tomography, union-free families, minimum hitting set, hypergraph transversal.
%\end{keywords}

\section{Introduction}

Identifying a subset of defective items out of a much larger set of items is a problem that found numerous application in a variety of situations such as medical screening, network reliability, DNA  screening, streaming algorithms. {\em Network Tomography} is a general inference technique based on end-to-end measurements aimed to extract internal network characteristics such as link delays and link loss rates but also defective items.
In this paper, we consider {\em Boolean Network Tomography} ($\BNT$) where the outcome of the measurements is a boolean value. Duffield, who as first introduced boolean network tomography \cite{DBLP:journals/tit/Duffield06} to identify network failure components, proposed an  inference algorithm based on $\BNT$ to identify sets of failure links. 
The $\BNT$ approach  was later studied also to identify node failures in networks  \cite{DBLP:conf/imc/MaHSTLL14,DBLP:journals/pe/MaHSTL15,DBLP:journals/ton/MaHSTL17,DBLP:conf/icdcs/GalesiR18, DBLP:conf/algosensors/GalesiRZ19, DBLP:journals/ton/BartoliniHAMTK20}.  

In the case of identifying failure nodes, the $\BNT$ approach deals with extracting as much accurate as possible information on the number and  the positions of the 
corrupted nodes  from the solutions $\vec x$ of a boolean  system $\pp \vec x = \vec b$, where  $\pp$  is the incidence matrix of the $m$ measurement paths over the $n$ nodes and $\vec b$ is the $m$-vector of the 
boolean  outcomes of the measurement paths (see Figure \ref{fig:example}). The challenge of localizing failure nodes is that different sets of failure nodes can produce the same  measurement along the paths and so are indistinguishable from each other using the measurements. 
This leads to pose the following  question: given the set of paths $\pp$ what is the maximal set of defective nodes we can hope
to identify unambiguously ?  Identifiability conditions on the matrix $\pp$ under which failure nodes can be localized  unambiguously (or  also counted accurately) from the  
solution of the system $\pp \vec x = \vec b$ are of course of the utmost interest. In this paper we study old and introduce new of these conditions,
contributing:
\begin{enumerate} 
\item to understand the combinatorics and the complexity of the theoretical problem of  unambiguously identify failure node sets under the $\BNT$ approach, and,  
\item to devise new algorithms and heuristics to count or localize as more precisely as possible failure nodes in networks.   
\end{enumerate}
\begin{figure*}
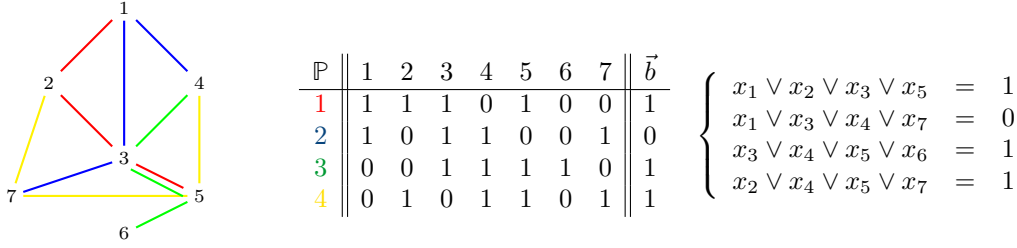

\begin{subfigure}{.3\textwidth}
  \centering
\begin{center}
\begin{scriptsize}
\tikz {
%\node (g) at  (-2,-1) {$\pp$};
\node (v1) at  (0,0) {$1$};
\node (v2) at  (-1,-1)  {$2$};
\node (v3) at  (0,-2) {$3$};
\node (v3f) at  (0,-2.1){};
\node (v4) at  (1,-1) {$4$};
%\node (v5) at  (-1,0) {$v_5$};
\node (v5) at  (1,-2.5) {$5$};
\node (v5f) at  (0.87,-2.55) {};
\node (v7) at (-1.5,-2.5){$7$};
\node (v6) at  (0,-3) {$6$};

%\draw[-,red, thick, midway, above, blue] (v1)--(v2);
\draw[-, red, thick] (v1)--(v2);
\draw[-,red,thick] (v2)--(v3);
\draw[-,bicolor={red}{green},thick] (v3)--(v5);

\draw[-,blue,thick] (v7)--(v3);
\draw[-,blue,thick] (v3)--(v1);
\draw[-,blue,thick] (v1)--(v4);

\draw[-,yellow,thick] (v4)--(v5);
\draw[-,yellow,thick] (v5)--(v7);
\draw[-,yellow,thick] (v7)--(v2);

\draw[-,green,thick] (v4)--(v3);
\draw[-,green,thick] (v3f)--(v5f);
\draw[-,green,thick] (v5)--(v6);
}
\end{scriptsize}
\end{center}
%\caption{A set of $\pp$ of 4 paths over $7$ nodes}
\end{subfigure}
\begin{subfigure}{.3\textwidth}
  \centering
 $$
 \begin{array}{c||ccccccc||c}
\pp & 1&2&3&4&5&6&7& \vec b \\
\hline
\red{1} & 1&1&1&0&1&0&0&1 \\
\blue{2} &1&0&1&1&0&0&1&0 \\
\greenblue{3} & 0&0&1&1&1&1&0&1 \\
\yellow{4} & 0&1&0&1&1&0&1&1\\
 \end{array}
 $$ 
% \caption{The incidence matrix of $\pp$ and the measurements vector $\vec b$}
\end{subfigure}
\begin{subfigure}{.3\textwidth}
  \centering
  $$
  \left \{\begin{array}{lll}
  x_1 \vee x_2 \vee x_3 \vee x_5 &=& 1\\
  x_1 \vee x_3 \vee x_4 \vee x_7 &=& 0\\
  x_3 \vee x_4 \vee x_5 \vee x_6 &=&1 \\
  x_2 \vee x_4 \vee x_5 \vee x_7 &=&1 
  \end{array}
  \right.
$$
\end{subfigure}

\caption{(1) A set $\pp$ of 4 paths over $7$ nodes. (2) The incidence matrix of $\pp$ and a measurements vector $\vec b$. 
(3) The associated boolean system.  Notice that the outcome  $1$ in the measurement of a path  indicates the presence of at least a node  failure.}
\label{fig:example}
\end{figure*}

\subsection{Previous work}
The condition introduced  as {\em $k$-identifiability} (for $\pp$) states that any two distinct node sets of size at most $k$ 
can be separated by at least a path in $\pp$. $k$-identifiability  initially introduced for link failure detection \cite{DBLP:journals/ton/MaHLST14,DBLP:conf/infocom/RenD16},  was later  studied with success also for node failure detection \cite{DBLP:conf/imc/MaHSTLL14,DBLP:journals/pe/MaHSTL15,DBLP:journals/ton/MaHSTL17,DBLP:conf/icdcs/GalesiR18, DBLP:conf/algosensors/GalesiRZ19, DBLP:journals/ton/BartoliniHAMTK20}. If this condition is true for a set of measurements paths $\pp$ it ensures  that 
if there are at most $k$ failure nodes in $\pp$ then these nodes can be identified unambiguously.  Hence the optimization problem of computing  
the maximal  $k\leq n$ such that a set $\pp$ is $k$-identifiable ($k$-$\ID$), i.e. admits the $k$-identifiability property is very relevant 
to the problem of node failure localization. We refer to this maximal value as $\mu(\pp)$ (it was called $\Omega(\pp)$ in \cite{DBLP:journals/pe/MaHSTL15}). 

As observed in \cite{DBLP:journals/pe/MaHSTL15,DBLP:journals/ton/MaHSTL17} $k$-identifiability can be {\em scaled} to each single node yet preserving the property for the whole set  of paths. A node $u$  is $k$-$\ID$ if any two sets of size at most $k$ differing on $u$ are separated by at least a path in $\pp$ (see Definition \ref{def:idnode}).  
Hence understanding the combinatorics  of the  set $\ID_k(\pp)$ of the $k$-identifiable nodes in $\pp$ and study upper and lower bounds for 
$|\ID_k(\pp)|$ is of great importance to develop algorithms to maximize the identification of failure nodes in real networks. 

Both definitions of $k$-identifiability were largely investigated. In \cite{DBLP:journals/ton/MaHSTL17} they started this study quantifying the capability of failure localization through (1) the maximum number of failures such that failures within a given node set can be localized unambiguously, and (2) the largest node set  failures
can be uniquely localized under a given bound on the total number of failures. These measures where used  to evaluate the impact of maximum identifiability  on various parameters of the network (underlying the  set of  paths) like the  topology, the number of monitor and the probing mechanisms. 

%They presented a set of sufficient and necessary conditions (testable in polynomial time)  for identifying a bounded number of failures within an arbitrary node set.   

In the work \cite{DBLP:conf/icdcs/GalesiR18,GR20}  we studied $k$-identifiability from the topological point of view of the graph underlying $\pp$. 
We were proving tight bounds on the maximum identifiability that can be reached in the case of topologies  like trees, 
grids and hypergrids and under embeddings on directed graphs. Our results culminated with a heuristic to design networks with a 
high degree of identifiability or to modify a network to boost identifiability.    
In the work \cite{DBLP:conf/algosensors/GalesiRZ19} we were employing {\em Menger's theorem} establishing a precise relation of  $\mu(\pp)$  with the vertex connectivity of the graph underlying $\pp$. We generalize results in \cite{DBLP:conf/icdcs/GalesiR18} to Line-of-Sight Networks and started the study of identifiability conditions on random graphs and random regular graphs.    

Monitor  placement can be in fact relevant to improve identifiability of failure nodes. The works  \cite{DBLP:journals/pe/MaHSTL15, DBLP:journals/ton/BartoliniHAMTK20}   considered  the problem of optimizing the capability of identifying network failures through different monitoring schemes  giving  upper bounds on the maximum number of identifiable nodes, given the number of monitoring paths, the routing scheme and  the maximum path length. 
In  \cite{DBLP:journals/ton/BartoliniHAMTK20} in particular studied upper bounds on the set of $|\ID_1(\pp)|$ and as in the case of \cite{DBLP:conf/icdcs/GalesiR18} they provide hueristics on how to design topologies and related monitoring schemes to achieve the maximum identifiability under various network settings.  

 \subsection{Contributions}
In this work we introduce new identifiability measures and deepen the study of $k$-identifiability obtaining several new results and new 
heuristics to test networks against the number and position of failing nodes.    From a theoretical  perspective  our contributions are the following
\begin{enumerate}
\item \label{item1} We set a question implicitly left open in some previous works \cite{DBLP:journals/ton/MaHLST14,DBLP:journals/ton/BartoliniHAMTK20} about  the limits of upper bounds on identifiability of node failures  via Boolean network tomography. %What is the minimal number $\kappa(m,k)$ of nodes such that $\pp$ is no longer $k$-identifiable, that is if $n\geq \kappa(m,k)$, then $\mu(\pp)<k$ ?
%of measurement paths $m$ on $n$ nodes such that  below $m$ $\pp$ is not  $k$-identifiable ?  
What are the precise tradeoffs between number of nodes $n$ and number of paths $m$ of $\pp$  such that $\pp$ is no longer $k$-identifiable, that is  $\mu(\pp)< k$? The answer is known only for $k=1$ where  the tradeoff $n \geq 2^{m}-1$ implies $\mu(\pp)<1$, is obtained by a straightforward counting argument ( see Lemma \ref{lem:mu1} and \cite{DBLP:journals/ton/BartoliniHAMTK20}).   

Using the notion of {\em regular union-free} families, we answer to the  problem for any $2 \leq k\leq n$, showing  that  $n \geq   2^{\frac{k}{k-1}(m+k-1)^{(1+\epsilon)}}$ implies $\mu(\pp)<k$, for any $\epsilon>0$.  

\item We introduce two new identifiability notions, namely, {\em $k$-separability} ($k$-$\SEP$) and {\em $k$-distinguishability} ($k$-$\DIS$).
Analogously to identifiability we define these notions on nodes and we consider the corresponding node sets $\SEP_k(\pp)$ and $\DIS_k(\pp)$.
These conditions provide significant upper and lower bounds to identifiability: namely we prove that  for all $k\leq n$, 
%$\SEP_k(\pp) \subsetneq \ID_k(\pp)$ (i.e.
 $k$-$\SEP$ {\em implies} $k$-$\ID$ and  
 %$\ID_k(\pp) \subsetneq \DIS_k(\pp)$ (i.e.  
 $k$-$\ID$ {\em implies} $k$-$\DIS$, both strictly. Hence $\SEP_k(\pp)\subseteq \ID_k(\pp) \subseteq \DIS_k(\pp)$. 
  We use these measures to get  upper and lower bounds for $|\ID_k(\pp)|$ and $\mu(\pp)$, to study the computational complexity  of  identifiability conditions and to estimates the number of $k$-identifiable nodes 
through a  random model. Namely:

\item \label{item2.b}We prove that  the problem of deciding the non $k$-separability (hence the non $k$-identifiability) of a given node in $\pp$ is polynomial time reducible to the  {\em minimum hitting set} problem ($\MHS$).  Furthermore we prove that the optimization problem of {\em finding the minimal $k$ such that a given node is not $k$-separable} in $\pp$  is $\NP$- complete.  To our knowledge these are the first known hardness results of identifiability problems arising from boolean network tomography. 
% The relation with the hitting set problem (see below) suggests the use of an algorithm to compute  upper bounds on $|\ID_k(\pp)|$ and on $\mu(\pp)$ based in the hypergraph transversal problem, which  we explain in Section \ref{sec:hypergraph}.  

\item \label{item3} We introduce and study a  random model for $\pp$  based on the binomial distribution and we estimate lower bounds  on the number of $k$-identifiable nodes $|\ID_k(\pp)|$ in this model by analyzing the number of $k$-separable nodes in $\pp$.

\item  \label{item2.a} We use node distinguishability to study upper bounds on the number of $k$-identifiable nodes parameterizing the search of such nodes
 in terms of specific subset of nodes and specific subset of paths in $\pp$.  We introduce the relation  (Definition \ref{def:keq})
 {\em $u$  $k$-equal $W$  modulo ${\cal P}$}, where $u$ is node,  $W$ a set of nodes and  ${\cal P}$ a family of paths in $\pp$
 that characterizes non-distinguishability of $u$ restricted to the set  $W$ with respect to ${\cal P}$. 
 A recursive construction (Definition \ref{def:tau} of $\tau_k$) built on the previous relation allows to  upper bound efficiently the number of $k$-identifiable nodes in a fine-grained way.

\end{enumerate} 

From a more applied perspective our results have the following consequences and applications.
\begin{enumerate}
\item  The  result in item \ref{item1} can be used as an  estimate of upper bounds on the  number of $k$-identifiable nodes in $\pp$. As \cite{DBLP:journals/ton/BartoliniHAMTK20}  use the  result for $k=1$   to prove that $|\ID_1(\pp)|\leq \min(n,2^m-1)$ (see Theorem \ref{thm:arrigo}), our  bound  proves the general statement that for all $2\leq k\leq n$,  $|\ID_k(\pp)| \leq \min\{n,2^{\frac{k}{k-1}(m+k-1)^{(1+\epsilon)}}\}$ (Theorem \ref{thm:mainbound}).  Our bound can also be 
used as a black-box in algorithms and heuristics aimed at approximating the number of identifiable (\cite{DBLP:conf/imc/MaHSTLL14, DBLP:journals/pe/MaHSTL15,DBLP:journals/ton/MaHSTL17}) nodes which use the bound for $k=1$. For instance the ICE heuristic of \cite{DBLP:journals/ton/BartoliniHAMTK20}, that creates a set of paths $\pp$ reaching a certain value of $\mu(\pp)$,  is generating paths according to  the result for $k=1$. 

\item The fact that the $\MHS$  problem is reducible to the non-separability problem (item \ref{item2.b}) suggests the idea of using the {\em minimal hypergraph transversal} (instead of a minimum hitting set) to lower bound the number of separable nodes (hence identifiable nodes) in $\pp$. Given  an order of the  variables a minimum hypergraph transversal in a set-system can be efficiently computed. We propose two algorithms based on the hypergraph transversal ({\tt Simple}-$\SEP$ and {\tt Decr}-$\SEP$). In particular in the second algorithm we use a new idea which partition the set of nodes of  $\pp$ in family of subsets of nodes called {\em $0$-decreasing} which allow to apply in a more efficient way  the hypergraph transversal heuristic ({\tt Decr}-$\SEP$). 

\item We employ the random model in item \ref{item3} to approximately counting  the number of $k$-identifiable nodes on concrete networks using an approach  based on the {\em maximum likelihood estimate} for binomial distributions.  Our experimental results indicate that  a lower bound for the number of 
$k$-identifiable nodes of a real network can be computed very accurately  using a relatively simple random model based on the binomial  distributions and computing the probability that a node is $k$-separable in this model. We  then consider a real set of measurement paths $\hat \pp$ as it was a random experiment, we plug in the MLE estimates on $\hat \pp$ in  the probability formula of the random model to estimate the cardinality of the sets $\SEP_k(\hat \pp)$.   
\item We use the definition of $\tau_k$  to  upper bound  the number of $k$-identifiable nodes in $\pp$ according to specific families of subset of nodes and subset of paths. As we show in Section \ref{sec:app}, this can be used to compute approximations of the value of $\mu(\pp)$ and $|\ID_k(\pp)|$ which are efficiently computable (Algorithm {\tt lb-$\DIS_k$}).  

\end{enumerate}

The paper is organized as follows: first we give the preliminary definitions on boolean network tomography and identifiability, showing the connection with unambiguous identification of failure nodes. In Section \ref{sec:limit} we study the tradeoffs between number of nodes and number of paths. In Section \ref{sec:ref} we give the definitions of $k$-separability and $k$-distinguishability and prove the relation with identifiability. In Section \ref{sec:random} we introduce the   
random model and we show how to count $k$-separable nodes (hence lower bounds on  $k$-identifiable nodes) on real  networks through a  {\em maximum likelihood estimate} method.
 In Section \ref{sec:hypergraph} we present the 
results on the computational complexity of $k$-identifiability and we introduce two algorithms based on hypergraph transversal to count identifiable nodes.  In Section \ref{sec:app} we introduce the definition of $\tau_k$  and a corresponding method (based on distinguishability) {\tt } to compute upper bounds on  identifiable nodes in a fine-grained way, when the set of paths is obtained by taking all the paths in a graph from a set of sources to a set of target nodes.

\section{Preliminary definitions}
%\red{Intro to node failure identifiability matrix, linear system}

Let $n,k \in \mathbb N $ and $k\leq n$. ${[n]  \choose k}$ is the set of subsets of $[n]$ of size $k$. ${[n]  \choose \leq k}$ be the set of subsets of $[n]$ of size at most $k$.
$2^{A}$ is the set of subsets of the set $A$. $A\oplus B$ is the symmetric difference between $A$ and $B$. $\overline A$ denotes the complement of $A$.

Let $n$ and $m$ be positive integers. We encode a  {\em set of $m$ paths over nodes in $[n]$} as a collection $\pp$  of $n$ {\em distinct} $m$-bit vectors
such that $ \mathbb 0 \not \in \pp$, i.e. the $m$-bit zero vector is not in $\pp$ (this condition means that each node in $[n]$ is used in at least a path).  

We can view $\pp$ in three different ways: as a boolean $m \times n$-matrix, as a collection of $n$ $m$-bit vectors  and  as a collection
of $m$ $n$-bit vectors.  For a node $u \in [n]$, $\co_u$ is then the  $m$-bit vector whose $p$-th coordinate indicates whether the node $u$ is in the $p$-th path or not.

We use also $\pp$ in a graph notation as follows:  if $u\in [n]$ is a node, then $\pp(u)$ identifies the set of all paths touching $u$, in other words the set
$\{p \in [m] : \pp[p,u]=1\}$. If $U \subseteq [n]$ is a set of nodes, $\pp(U)$ denotes the set of paths in $[m]$ touching at least a node in $U$, i.e.
$\pp(U)=\bigcup_{u\in U} \pp(u)$.

\subsection{Identifiability}

Let $\pp$ be a set of $m$ paths over $n$ nodes. We consider the following definition \cite{DBLP:conf/imc/MaHSTLL14}
\begin{definition}
\label{def:kid}
$\pp$ is $k$-identifiable  if  for all $U,W \subseteq [n]$ such that $|U|,|W| \leq k$ and $U \not = W$, it holds that $\pp(U) \not = \pp(W)$. 
\end{definition}

Notice that in terms of the column-vector notation, the previous definition says that  
for all distinct sets $U,W\subseteq [n]$ of size at most $k$, $\displaystyle \bigvee_{u\in U} \mathbb c_u \oplus \bigvee_{w \in W} \mathbb c_{w} \not =\mathbb 0 .
$

The definition of $k$-identifiability can be equivalently given for nodes $u\in [n]$ as follows (see also \cite{DBLP:journals/pe/MaHSTL15}).
\begin{definition} ($k$-identifiable nodes)
\label{def:idnode}
A node $u\in [n]$ is  $k$-identifiable with respect to $\pp$, if  for all $U,W \subseteq [n]$ of size at most $k$ and such that  $U\cap\{u\} \not = W\cap \{u\}$, 
it holds that $\pp(U) \not = \pp(W)$. 
\end{definition}

$\ID_k(\pp)$ denotes the set of $k$-identifiable nodes in $\pp$. From the definitions a second of thought allows to see that  $k$-identifiability implies $k'$-identifiability
for $k'<k$. Hence
\begin{lemma} Let $\pp$ be a set of $m$ paths over $n$ nodes. 
Then $\ID_k(\pp) \subseteq \ID_{k'}(\pp)$ for $k'\leq k\leq n$.
\end{lemma}

Furthermore scaling to identifiability of nodes does not affect the main property of $k$-identifiability (which we see below).
Next theorem is proved in \cite{DBLP:journals/pe/MaHSTL15} (Theorem 4). 
 \begin{theorem}(\cite{DBLP:journals/pe/MaHSTL15}) Let $\pp$ be a set of $m$ paths over $n$ nodes. 
 $\pp$ is $k$-identifiable if and only if every node in $[n]$ is $k$-identifiable with respect to $\pp$.
 \end{theorem}

We denote by $\mu(\pp)$ the maximal $k\leq n$ s.t.  $\pp$ is $k$-identifiable.

\medskip

Let us  motivate our definitions in the context of the approach  of boolean network tomography to detect failure nodes in networks.
Assume to have  a set $\pp$ of $m$ end-to-end paths over $n$ nodes $\pp$. 
A {\em binary measurement} $\mathbb M$ along a path $p \in [m]$ is obtained by sending a message through $p$ and 
recording the outcome $\mathbb M(p)$, a bit,  which identifies 
(in the case $\mathbb M(p)=1$) that some node in $p$ is failing, or (in the case $\mathbb M(p)=0$) that no node is failing along the path $p$.

We claim that if $\pp$ is $k$-identifiable, then under any binary measurement $\mathbb M$, 
we can {\em uniquely localize} in $\pp$ up to $k$ failing nodes. 
Given a binary measurement $\mathbb M$ over $\pp$, let $\fail_{\mathbb M}(\pp) = \{p\in [m] | \mathbb M(p)=1\}$.

\begin{definition}(Unique failure)
Let $\pp$ be a set of $m$ paths over $n$ nodes and $\mathbb M$ a binary measurement on  $\pp$. 
A set of nodes $W\subseteq[n]$ is {\em failing} in $\pp$ if
$ \pp(W) \subseteq  \fail_{\mathbb M}(\pp)$. 
 $W$ is {\em uniquely failing} if  it is failing and furthermore $\overline{\pp(W)} \subseteq \overline{\fail_{\mathbb M}(\pp)}$, i.e. on any path
not touching $W$ the measurement is not failing.
\end{definition}

%\begin{definition}(uniquely failing set of  nodes)
%Let $\pp$ be a set of $m$ paths over $n$ nodes and $\mathbb M$ a binary measurement along all paths  in $\pp$. 
%A set of nodes $W\subseteq[n]$ is {\em uniquely failing} in $\pp$ if $W$ is failing in $\pp$ and
%for any path $p \in [m]$ touching only nodes not in $W$ it holds that $\mathbb M(p)=0$. 
%\end{definition}

\begin{theorem}
Let $\pp$ be a set of $m$ paths over $n$ nodes. If $\mu(\pp)\geq k$, then there is exactly one set  
of nodes   of size at most $k$ that is uniquely failing in $\pp$.   
\end{theorem}

\begin{proof}
Let $\mathbb M$ be a measurements over $\pp$. Assume by contradiction that there are two distinct sets $U$ and $W$ of size at most $k$  
which are  both uniquely failing in $\pp$ under $\mathbb M$.
Since $U\not =W$, and $\pp$ is $k$-identifiabale, then there is either a $p \in \pp(U)\setminus \pp(W)$ or a $p\in \pp(W)\setminus \pp(U)$.
Assume wlog the former. Since $U$ is failing, then  $\mathbb M(p)=0$. But since $W$ is uniquely failing and $p \in \overline{\pp(W)}$, then  $p \in \overline{\fail_{\mathbb M}(\pp)}$ and hence
$\mathbb M(p)=1$. Contradiction. 
\end{proof}

\section{Upper bounds on $\mu(\pp)$ by counting}
\label{sec:limit}

%\subsection{Upper bounds by cardinality argument}
In this subsection we show that under what bounds on the number of paths $m$ in $\pp$, we have that  $\mu(\pp)<k$.  
We start by showing under what conditions on $m$, $\mu(\pp)<1$. 

Notice that to prove that $\mu(\pp)<1$, by Definition \ref{def:kid}  it is sufficient to find
two distinct nodes $u,w \in [n]$ such that $\co_u\oplus \co_w =  \mathbb 0$, that is for all $p \in [m] : \co_u[p] = \co_w[p]$. 
%First we give a simple counting argument which shows that if in $\pp$  the number $m$ paths are too smaller compared with the number of nodes $n$ then  $\mu(\pp) <1$. 
$\mu(\pp)<1$ will follow from a easy information theoretic bound on sets of $m$-vectors.

\begin{lemma}
\label{lem:mu1}
Let $\pp$ be a set of $m$ paths built on $n$ nodes. If $m < \log_2 (n+1)$, then $\mu(\pp)<1$. 
\end{lemma}
\begin{proof}
 $\pp$ is a collection of $n$ $m$-bit strings. There are at most $2^m-1$ different such strings ($\mathbb 0 \not \in \pp$).  Hence whenever $n>2^m-1$ there are two elements  $u\not =w \in [n]$ such that $\co_u=\co_w$, which  means $\co_u \oplus \co_w =  \mathbb 0$. 
\end{proof}

Corollary IV.1 in \cite{DBLP:journals/ton/BartoliniHAMTK20} can be obtained by previous observation immediately.

\begin{theorem} (\cite{DBLP:journals/ton/BartoliniHAMTK20})
\label{thm:arrigo}
Let $\pp$ be a set of $m$ paths over $n$ nodes. Then
$|\ID_1(\pp)| \leq \min \{n,2^m-1\}$.
\end{theorem}
\begin{proof}
$|\ID_1(\pp)|\leq n$ since it is a set of nodes. Assume that $n > 2^m -1$, hence by previous Lemma \ref{lem:mu1}
$\mu(\pp)=0$, hence there are at least two nodes $u_1\not = u_2$ not $1$-identifiable. Hence $|\ID_1(\pp)|\leq 2^m-1$ 
\end{proof}

We will prove a similar results for $\mu(\pp)<k$ for a generic $k\leq n$. %We start by the case $k=2$.
\subsection{Union-free families and upper bounds for $k$-identifiability}

A hypergraph $\mathcal{F}$ on the set $[m]$ is a family of distinct subsets of $[m]$, called edges of $\mathcal{F}$. 
If each edge is of fixed size $r\leq m$, then $\mathcal{F}$ is said to be {\em $r$-regular}, i.e., $\mathcal{F} \subset {[m] \choose r}$. 

\begin{definition}
\label{def:unionfreefamily}
For a positive integer $k$, $\mathcal{F}$ is called {\em $k$-union-free} if for any two distinct subsets of edges $\mathcal{A},\mathcal{B} \subseteq \mathcal{F}$, with $1 \leq |\mathcal{A}|,|\mathcal{B}| \leq k$, it holds that $\displaystyle \cup_{ A\in \mathcal{A}} A \not = \cup_{ B \in \mathcal{B}} B$. 
\end{definition}
Union-free regular  hypergraphs are  investigated in extremal combinatorics \cite{FF86}. 
It is immediate to see that a set $\pp$ of $m$ paths over $n$ nodes defines  a hypergraph ${\cal F}_\pp$ on the set $[m]$ in the following way: 
for $i\in [n]$ let $A_i =\{j\in [m] | \co_i[j]=1\}$ and define ${\cal F}_\pp=\{A_1,\ldots, A_n\}$. 
 Given a $U\subseteq [n]$, 
consider the subset of  ${\cal F}_\pp$, ${\cal U}=\{A_i \in {\cal F}_\pp | i \in U \}$. Observe that then $\pp(U)= \bigcup_{A \in {\cal U}} A $.
Hence immediately by definition of $k$-identifiability  and that of $k$-union-freeness it follows that:
\begin{lemma}
\label{lem:FP}
If $\pp$ is a set of $m$ paths over $n$ nodes and $\mu(\pp)\geq k$, then $ {\cal F}_\pp$ is k-union free. 
\end{lemma} 

${\cal F}_\pp$ is not necessary a regular hypergraph. For $r\in [m]$ let ${\cal F}_\pp(r)=\{A \in {\cal F}_\pp | |A| =r\}$.
Notice that each  ${\cal F}_\pp(r)$ is now a $r$-regular hypergraph on $[m]$. Moreover the family of the ${\cal F}_\pp(r)$'s partitions
${\cal F}_\pp$ and hence $|{\cal F}_\pp| = \sum_{r\in[m]}|{\cal F}_\pp(r)|$. Since  $|{\cal F}_\pp| =n$, it follows that:

\begin{lemma}
\label{lem:sum}
$\sum_{r\in[m]}|{\cal F}_\pp(r)|=n$. 
\end{lemma}

Furthermore notice that if ${\cal F}_\pp$ is $k$-union free then such it will be ${\cal F}_\pp(r)$ for each $r \in [m]$.

\smallskip

Let $m>r$, $k\in [m]$ with $k\geq2$,  and let $f(k,r,m)$ denote the maximum cardinality of a $k$-union-free $r$-regular hypergraph over $[m]$. 

%Let us extend this argument to a  generic $k\geq 1$. To prove that $\mu(\pp)<k$ it is sufficient to find two sets $U, W\subseteq [n]$, with $|U|,|W| \leq k$ such that 
%$\displaystyle \bigvee_{u\in U} \mathbb c_u \oplus \bigvee_{w \in W} \mathbb c_{w} =\mathbb 0 .$
\begin{theorem} [\cite{FF86,ST20}]
\label{th:2unionfree}
%There exist positive constants $C_r, C'_r$, such that 
$\Omega(m^{\frac{r}{k-1}})\leq f(k,r,m) \leq O(m^{\lceil \frac{r}{k-1}\rceil }).$
\end{theorem}

Let $m_0 \in \mathbb N$  and $C$ be the constant such that for all $m \geq m_0$, $f(k,r,m) \leq C m^{\lceil \frac{r}{k-1}\rceil }$.

\begin{theorem} 
\label{thm:uf}
Let $m$ be an integer such that $m \geq m_0$. Let $\pp$ be a set of $m$ paths over $n$ nodes. If 
$n> \sum_{r \in [m]} C m^{\lceil \frac{r}{k-1}\rceil }$, then  $\mu(\pp) < k$.
\end{theorem}

\begin{proof} Assume by contradiction  that   $n> \sum_{r \in [m]} C m^{\lceil \frac{r}{k-1}\rceil }$ and  $\mu(\pp) \geq  k$.
By Lemma \ref{lem:FP}  ${\cal F}_\pp$ is $k$-union free. Hence (see observation after Lemma \ref{lem:sum}) for each 
$r\in [m]$, ${\cal F}_\pp(r)$ is a $r$-regular $k$-union free hypergraph and hence  by previous theorem $|{\cal F}_\pp(r)| \leq C m^{\lceil \frac{r}{k-1}\rceil }$.
The ${\cal F}_\pp(r)$ partition ${\cal F}_\pp$ and by Lemma \ref{lem:sum}  we have 
$n=\sum_{r\in[m]}|{\cal F}_\pp(r)| \leq \sum_{r \in [m]} C m^{\lceil \frac{r}{k-1}\rceil }$.
\end{proof}

\begin{corollary} 
\label{cor:mainup}
Let $\pp$ be a set of $m$ paths over $n$ nodes and $2 \leq k \leq m$. If $m< \sqrt[1+\epsilon]{\frac{(k-1)}{k}(\log_2 n-D)}-(k-1)$, for some $\epsilon >0$ and where 
$D=\log C$, then $\mu(\pp)<k$.
\end{corollary}
\begin{proof}
Assume for the moment that $m$ divides $k-1$. We prove that if $m< \sqrt[1+\epsilon]{\frac{(k-1)}{k}(\log_2 n-D)}$, then 
\begin{eqnarray}
n> Cmm^{\frac{m}{k-1}} \label{eq:1}.
\end{eqnarray}
 
This immediately implies   $n> \sum_{r \in [m]} C m^{\lceil \frac{r}{k-1}\rceil }$, since  $\sum_{r \in [m]} C m^{\lceil \frac{r}{k-1}\rceil } \leq Cmm^{\frac{m}{k-1}}$.
Equation \ref{eq:1} follows by the following implications:

\begin{eqnarray}
m< \sqrt[1+\epsilon]{\frac{k-1}{k}(\log_2 n-D)} \\
m^{1+\epsilon} < \frac{k-1}{k}(\log_2 n-D) \\
\frac{k}{k-1}m^{1+\epsilon} < \log_2 n -D \\
\log C + \frac{k}{k-1}m\log m<\log n  \label{eq: 5} \\
\log C + \log m + \frac{m}{k-1}\log m < \log n \label{eq: 6}\\
Cmm^{\frac{m}{k-1}}< n
\end{eqnarray}

Equation \ref{eq: 6} follows from Equation \ref{eq: 5} since $K= \frac{k}{k-1}>1$ and $K m \log m > \log m+\frac{m \log m}{k-1}$ for all $m$.

\smallskip

If $m$ does not divide $(k-1)$, let $a<(k-1)$ be the smallest non-negative integer such that $m+a$ divides $k-1$. 
Hence $m+a< m+(k-1)$. Let $\hat m =m+a$. Since $m< \sqrt[1+\epsilon]{\frac{(k-1)}{k}(\log_2 n-D)}-(k-1)$, then
$\hat m< \sqrt[1+\epsilon]{\frac{(k-1)}{k}(\log_2 n-D)}$. Hence the previous argument proves that 
$n> C\hat m\hat m^{\frac{\hat m}{k-1}}$. Since $m<\hat m$, then $n> Cmm^{\frac{m}{k-1}}$. 
Now by Theorem \ref{thm:uf}, this implies that $\mu(\pp)<k$.
\end{proof}

\begin{theorem}
\label{thm:mainbound}
Let $\pp$ be a set of $m$ paths over $n$ nodes. Then for all $k\leq n$, 
$|\ID_k(\pp)| \leq  \min\{n,2^{\frac{k(m+2k-2))^2}{k-1}}\}$.
\end{theorem}
\begin{proof}
Notice that if $n\geq 2^{\frac{k(m+k-1)^2}{k-1}}$, then $m< \sqrt[1+\epsilon]{\frac{(k-1)}{k}(\log_2 n-D)}-(k-1)$. Hence Corollary \ref{cor:mainup} and the same proof of Theorem \ref{thm:arrigo} imply the claim.
\end{proof}

\section{Refining identifiability: separability and distinguishability}
\label{sec:ref}
We introduce two new definitions approximating identifiability from above and from below  that we are going to use to  prove upper and lower bounds on the number of 
 $k$-identifiable nodes.

\begin{definition} ($k$-separable nodes)
A node $u\in [n]$ is  {\em $k$-separable} in $\pp$, if for all $U \subseteq [n]$ of size at most $k$ and such that $u \not \in U$, it holds that  there 
is a path $p \in \pp(u) \setminus \pp(U)$, i.e. there is at least a path passing though $u$ but not touching any node of $U$.
\end{definition}
We say that {\em $\pp$ is $k$-separable} if each node $u \in [n]$ is $k$-separable.  $k$-separability is a stronger notion than $k$-identifiability as captured by the following lemma.

\begin{lemma}
\label{lem:sepvsid}
If $u$ is  $k$-separable in $\pp$, then $u$ is is $k$-identifiable in $\pp$.
\end{lemma}
\begin{proof}
Let $U$ and $W$ be distinct subset of $[n]$ of size at most $k$. Then there exists a $u$  
such that  $U\cap\{u\} \not = W\cap \{u\}$ and then either  $u \in U\setminus W$ or $u \in W\setminus U$. Assume wlog the former. Then 
$u \not \in\ W$. $u$ is $k$-separable in $\pp$, there is a path $p\in \pp(u)\setminus \pp(W)$. Since $u\in U$, then $p\in \pp(U)\setminus \pp(W)$ 
and then $\pp(U)\not = \pp(W)$.
\end{proof}

Notice that opposite direction is not true as we argue: assume that $\pp$ is $k$-identifiable and that $u \not \in W$ for $W$ a set of at most 
$k$ nodes. The $k$-identifiability of $\pp$ implies that $\pp(u)\not = \pp(W)$, yet this condition alone does not guarantee that the path separating
$\{u\}$ from $W$, pass through $u$ and not touching $W$.

\begin{definition} ($k$-distinguishable nodes)
A node $u\in [n]$ is  {\em $k$-distinguishable} in $\pp$, if for all $U \subseteq [n]$ of size at most $k$ and such that $u \not \in U$, it holds 
$\pp(u) \not =  \pp(U)$.
\end{definition}
We say that {\em $\pp$ is $k$-distinguishable} if each node $u \in [n]$ is $k$-distinguishable. 

\begin{lemma}
If $u$ is is $k$-identifiable in $\pp$, then $u$ is  $k$-distinguishable in $\pp$.
\label{lem:idvsdis}
\end{lemma}
\begin{proof}
Assume that $u\in [n]$ is $k$-$\ID$ in. $\pp$.  Let  $W\subseteq [n]$ of size at most $k$ such that $u \not \in W$. We want to prove that $\pp(u) \not =\pp(W)$. 
By $k$-$\ID$ of $u$ we know that for all $U'$ and $W'$ in $[n]$ of size at most $k$ such that $U'\cap \{u\} \not = W'\cap \{u\}$, it holds that
$\pp(U') \not =\pp(W')$. Fix $U'=\{u\}$ and $W'=W$. Since $u\not \in W$, then $U'\cap \{u\} \not = W'\cap \{u\}$, hence
$\pp(u) \not =\pp(W)$, as required.
\end{proof}

Notice that the opposite direction is not necessary true: indeed if $u\in U\setminus W$,  knowing that $\pp(u) \not = \pp(W)$ it is not sufficient to conclude
$\pp(U)\not =\pp(W)$, exactly in those case when $\pp(u) \not = \pp(W)$ is witnessed by a path in $\pp(W)\setminus \pp(u)$, which can touch other nodes in $U$ but not $u$.
 
 \smallskip

We denote by $\ID_k(\pp), \SEP_k(\pp), \DIS_k(\pp)$ the set of nodes which are respectively $k$-identifiable, $k$-separable and $k$-distinguishable in $\pp$.
And we use to say respectively that $u$ is $k$-$\ID$, $k$-$\SEP$ and $k$-$\DIS$ in $\pp$. 

By previous Lemmas and discussion it holds that
\begin{lemma}
\label{lem:sepiddis}
For all $k\in [n]$, $|\SEP_k(\pp)| \leq |\ID_k(\pp)| \leq |\DIS_k(\pp)|$. 
\end{lemma}
  Furthermore since the three properties are clearly 
antimonotone, it holds that;   
 \begin{lemma}
 For all $k \in [n]$, $\ID_k(\pp) \subseteq \ID_{k-1}(\pp)$, $\SEP_k(\pp) \subseteq \SEP_{k-1}(\pp)$, $\DIS_k(\pp) \subseteq \DIS_{k-1}(\pp)$
\end{lemma}
 
We denote by $\sigma(\pp)$ (respectively $\delta(\pp)$) the maximal $k\leq n$ s.t.  $\pp$ is $k$-separable (respectively $k$-distinguishable).
Hence we have $\delta(\pp)\leq \mu(\pp)\leq \sigma(\pp)$.

\section{Lower bounds on $\mu(\pp)$  by a random model }
\label{sec:random}

To study lower bounds on $\ID_k(\pp)$  (or on  $\mu(\pp)$) for  real set of paths we introduce a simple random model. We are given 
$m$ and $n$ natural numbers and $n$ real numbers $\lambda_i \in [0,1]$. The random set  $\pp$ of $m$ paths over $n$ nodes is obtained by taking  
independently $n$ binary strings of length $m$ such that the $i$-th string 
is  distributed according to the binomial distributions $\Bin(m,\lambda_i)$. That means that node $i \in [n]$  will be present on each path with probability 
$\lambda_i$ and absent with  probability $(1-\lambda_i)$.

Our approach to estimate  $|\ID_k(\pp)|$ is the following: 
\begin{enumerate}
\item  by Lemma \ref{lem:sepiddis},  $|\SEP_k(\pp)| \leq |\ID_k(\pp)|$.
\item For $u\in [n]$  we obtain $\nu_{n,m,\lambda}(u)=\Pr[u \in \SEP_k(\pp)]$. 
\item  Given a real set $\hat \pp$ of  $M$ paths on $N$ nodes,  we  consider  $\hat \pp$ to be  a random  experiment and from $\hat \pp$ we compute 
a {\em maximum likelihood estimate} $\hat \lambda_i$ for each of the $\lambda_i$.
\item We estimate $|\SEP_k(\hat\pp)| = \displaystyle \sum_{u \in [N]} \nu_{N,M,\hat \lambda}(u)$
 \end{enumerate}

%\subsection{Lower bounds }
Let $u \in [n]$ and $W  \in {[n]-\{u\}\choose \leq k}$. Let us say that $(u,W)$ is $\GOOD$ if  there is a path $p \in [m]$ such that $p \in \pp(u) \setminus \pp(W)$.
 $(u,W)$ is $\BAD$ if it is not  $\GOOD$. 

\begin{lemma} 
\label{lem:bad}
Let $u \in [n]$ and $W \subseteq [n]\setminus \{u\}$.
$\Pr[(u,W) \BAD]= \left(1-\lambda_u\prod_{w\in W}(1-\lambda_w)\right)^m.$
\end{lemma}
\begin{proof}
$(u,W) $ is $\BAD$ if and only if  for all $p \in [m]: (p(u) \rightarrow p(W))$. Then $\Pr[(u,W) \BAD]= \left( \Pr[(p(u) \rightarrow p(W))]\right)^m$. 
 The condition $p(u) \rightarrow p(W)$ is the same as
$\neg p(u) \vee \bigvee_{w \in W}p(w)$ which is the same as $\neg \left(p(u) \wedge \bigwedge_{w \in W} \neg p(w)\right)$.
$\Pr[p(u)]= \lambda_u$ and $\Pr[\neg p(w)]= (1-\lambda_w)$. Hence the claim.
 \end{proof}

Let $k\leq n$. and let  $S(k) = {[n-1] \choose \leq k}$.
\begin{theorem} 
\label{thm:mainprob}
Let $n,m, k\in \mathbb N$, $u \in [n]$, and $k\leq n$.  $\Pr[u \in \SEP_k(\pp)] = \displaystyle \prod_{W \in S(k)}\left (1-(1-\lambda_u\prod_{w\in W}(1-\lambda_w))^m\right).$
\end{theorem}

\begin{proof} Observe that $\Pr[u \in \SEP_k(\pp)] = \Pr[\mbox{ $u$ is $k$-$\SEP$ in $\pp$}]  = 
\Pr[\forall W \mbox{with  $u \not \in  \!W$ and $|W| \leq k$}: (u, W) \GOOD]$. 
By previous Lemma  $\Pr[(u,W) \GOOD]=1-(1-\lambda_u\prod_{w\in W}(1-\lambda_w))^m$. Hence the theorem follows.
 \end{proof}

Assume we have a set $\hat \pp$ of $m$ paths over $n$ nodes.  We consider $\hat \pp$ as a random experiment. 
The standard approach to compute an MLE estimate $\hat \lambda_i$ of the $\lambda_i$ in the case of binomial distribution is to 
compute $\hat \lambda_i$ as the zero of the polynomial obtained by the prime derivative of the function expressing the probability   
that the node $i$ touches $N_i$ paths  in $\hat \pp$.

Let $p_i=\Pr[ \mbox{ node }  i \mbox{ touches } N_i \mbox{ paths  in }  \pp]$. Since in $\pp$ the column $i$ is distributed accordingly to the
$\Bin(m,\lambda_i)$, then  $p_i={m \choose N_i}  \lambda_{i}^{N_i} (1- \lambda_{i})^{m-N_i}$ . We 
study $\frac{d}{d\lambda_{i}} p_i$ and compute $\hat \lambda_i$ by setting $\frac{d}{d\lambda_{i}} p_i=0$. It is easy to see that 
this happen for  $\hat \lambda_{i} = \frac{N_i}{M}$.
\begin{figure*}[]
\centering
\begin{subfigure}{.3\textwidth}
 \centering
  \includegraphics[width=0.70\textwidth]{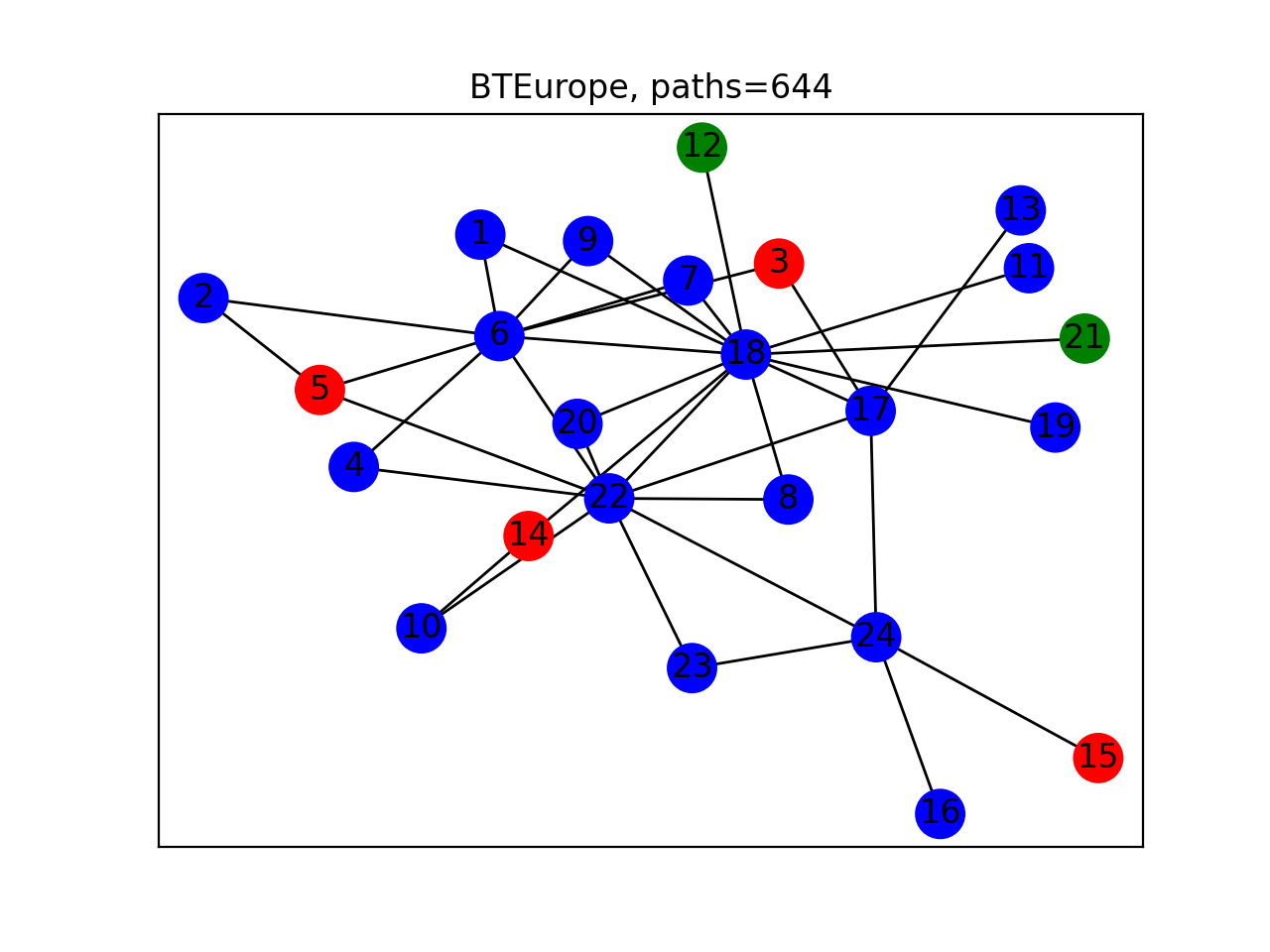}
  \end{subfigure}
  \begin{subfigure}{.3\textwidth}
  \centering
 	 \includegraphics[width=0.70\textwidth]{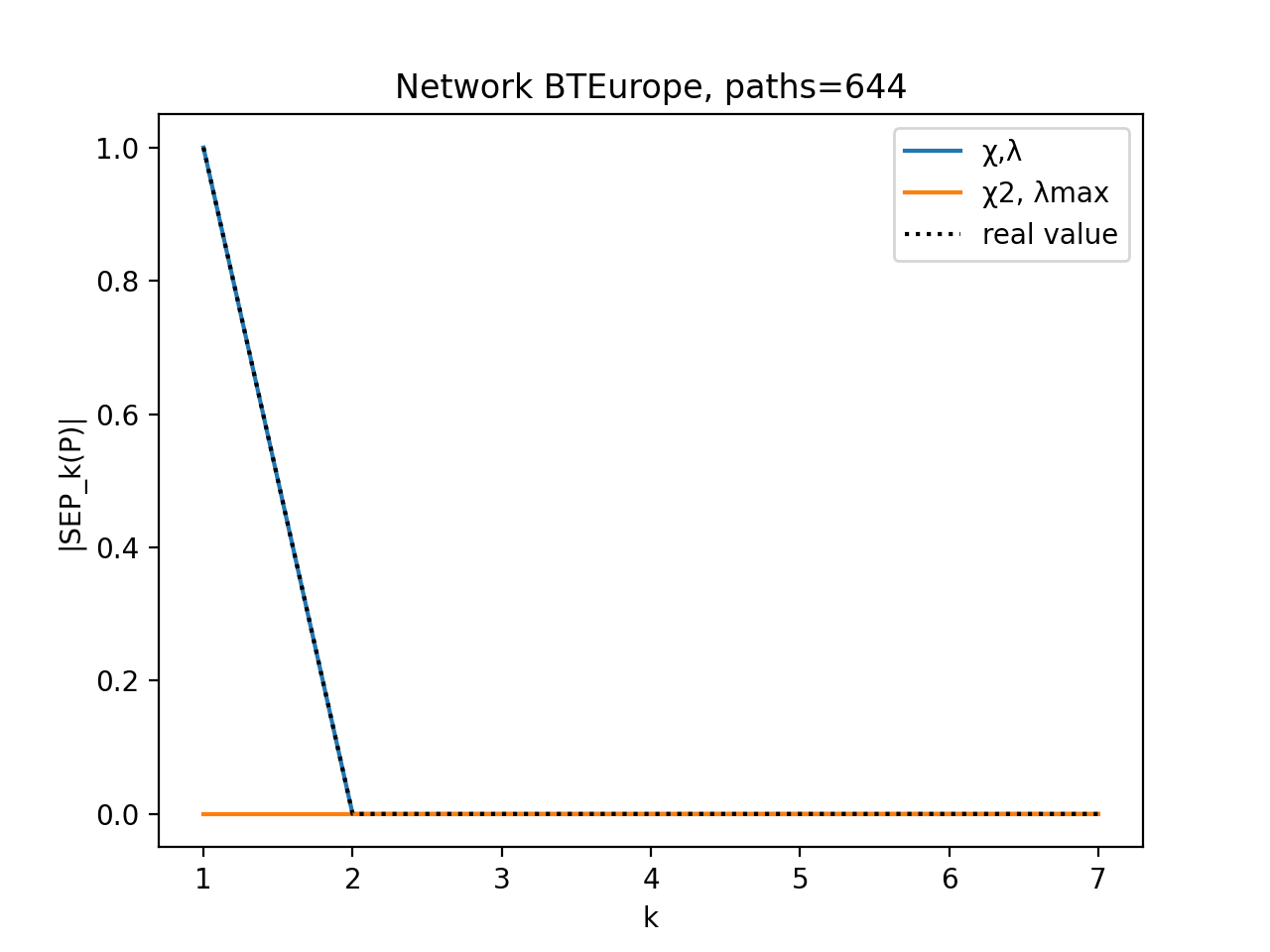}
  \end{subfigure}
  \begin{subfigure}{.3\textwidth}
  \centering
  \includegraphics[width=0.70\textwidth]{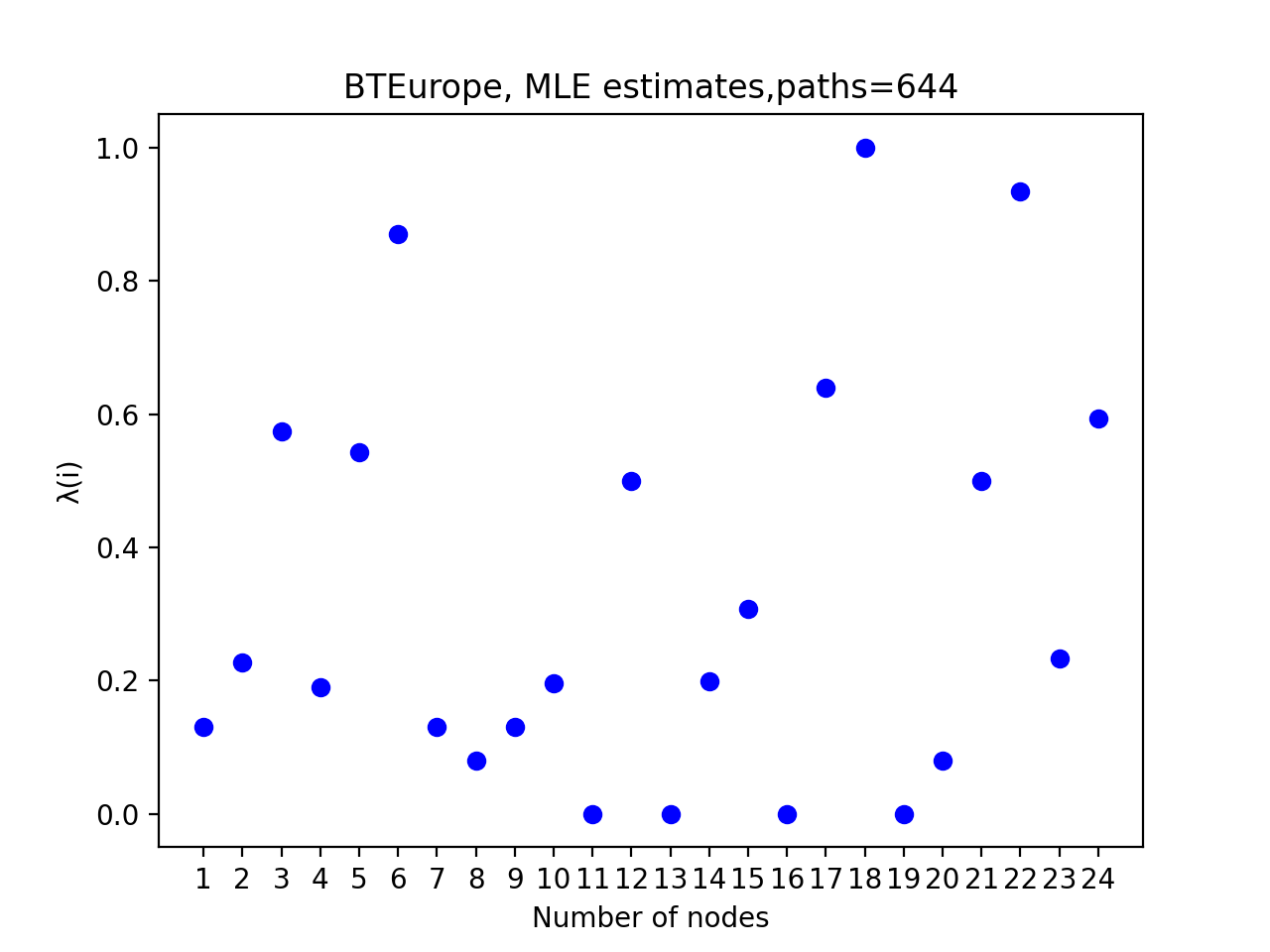}
  \end{subfigure}
 %\end{figure} 
\centering
 %\begin{figure}[H]
\begin{subfigure}{.3\textwidth}
  \centering
  \includegraphics[width=0.70\textwidth]{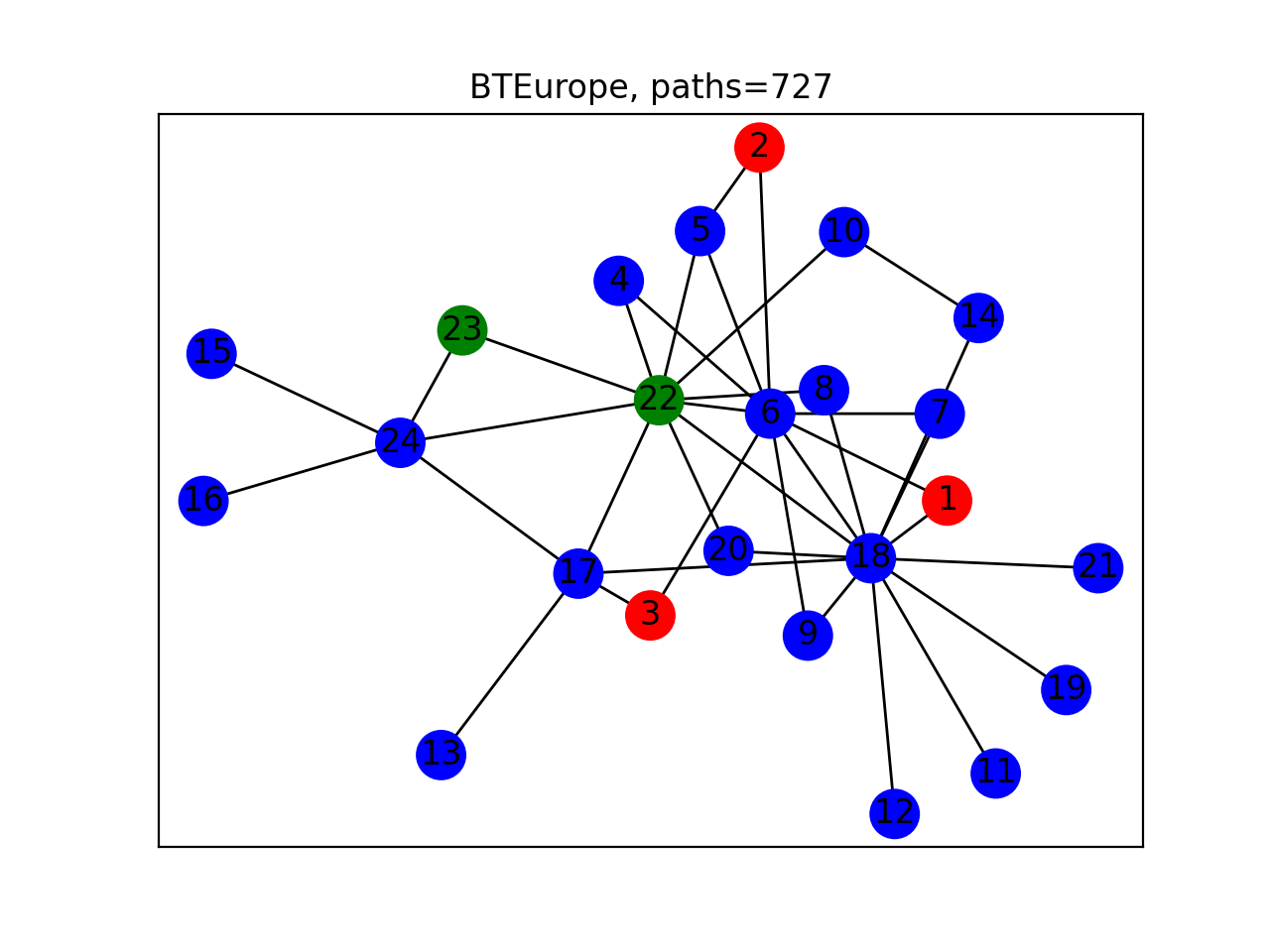}
  \end{subfigure}
  \begin{subfigure}{.3\textwidth}
  \centering
 	 \includegraphics[width=0.70\textwidth]{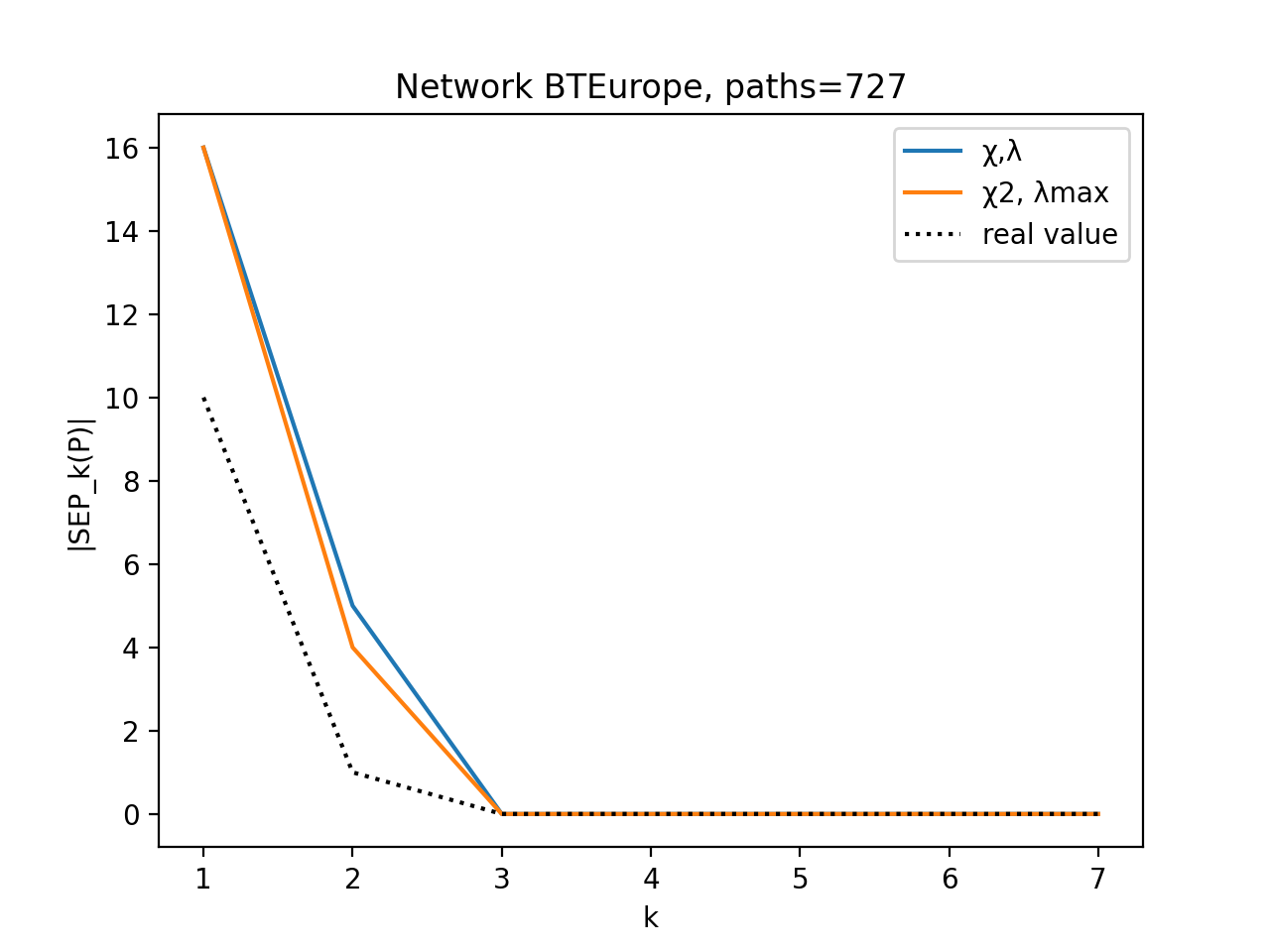}
  \end{subfigure}
  \begin{subfigure}{.3\textwidth}
  \centering
  \includegraphics[width=0.70\textwidth]{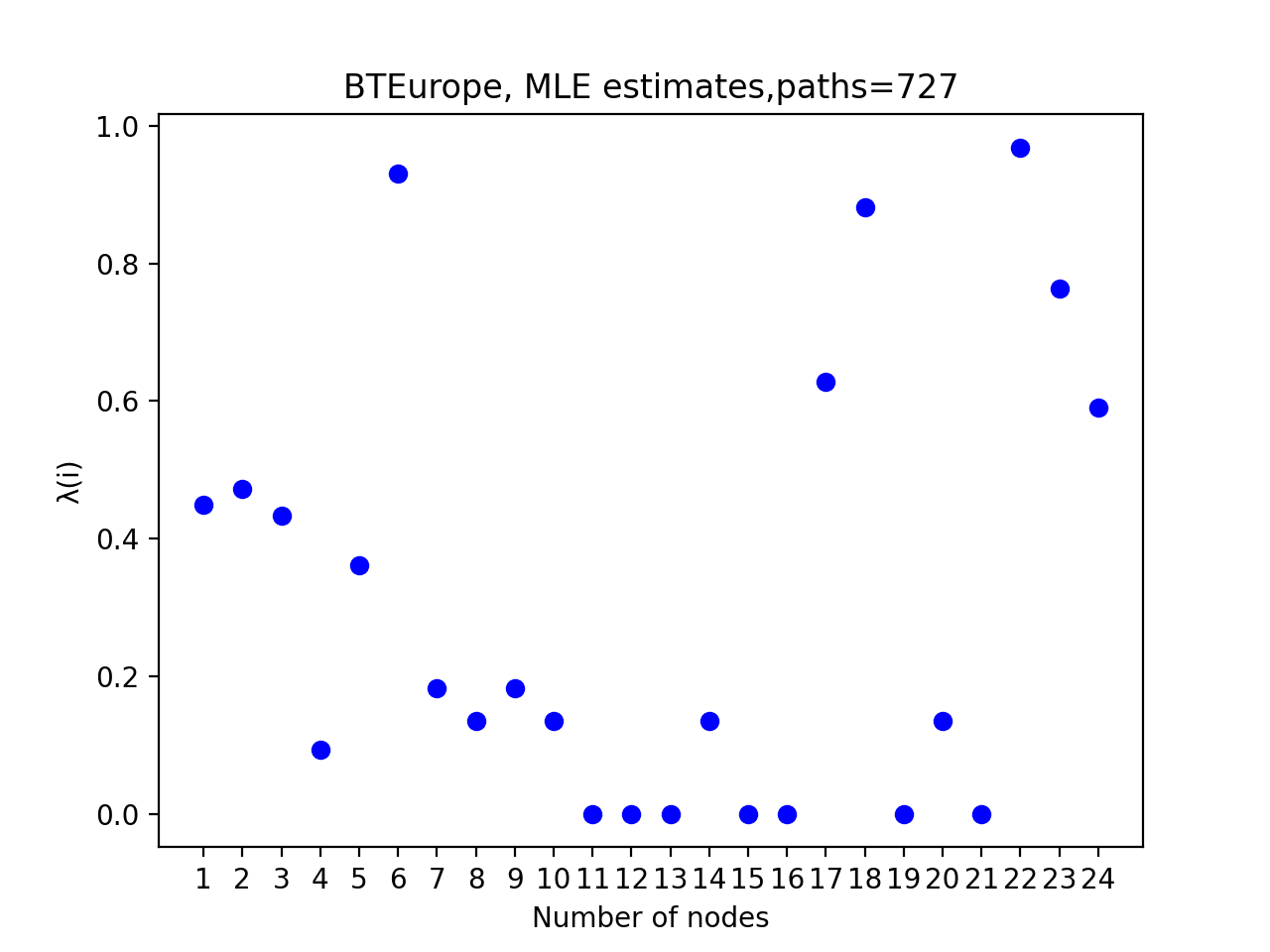}
  \end{subfigure}
  \caption{Data on the network BTEurope}
  \label{fig1}
   \end{figure*} 

 \begin{figure*}[]
\begin{subfigure}{.3\textwidth}
  \centering
  \includegraphics[width=0.70\textwidth]{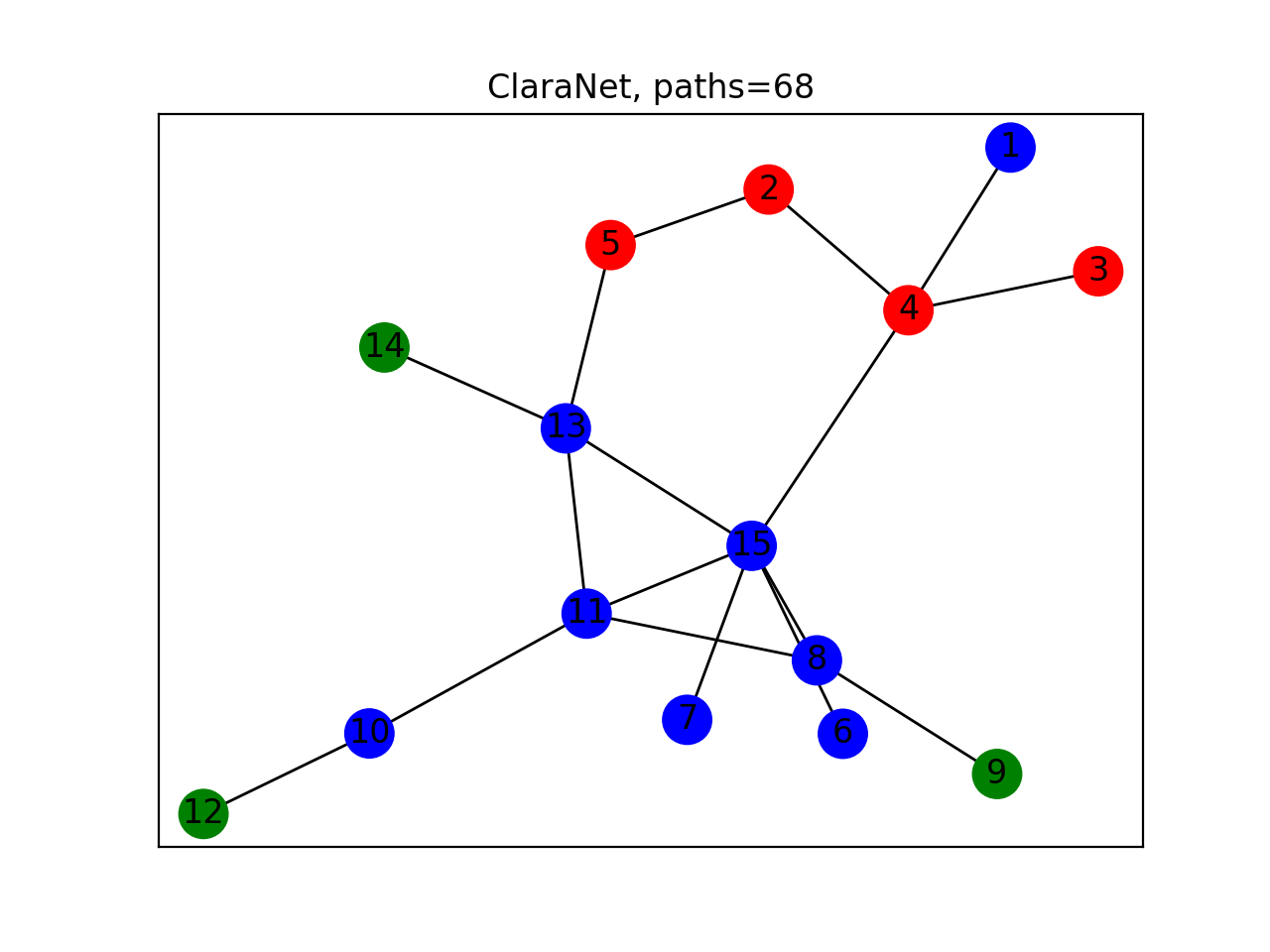}
  \end{subfigure}
  \begin{subfigure}{.3\textwidth}
  \centering
 	 \includegraphics[width=0.70\textwidth]{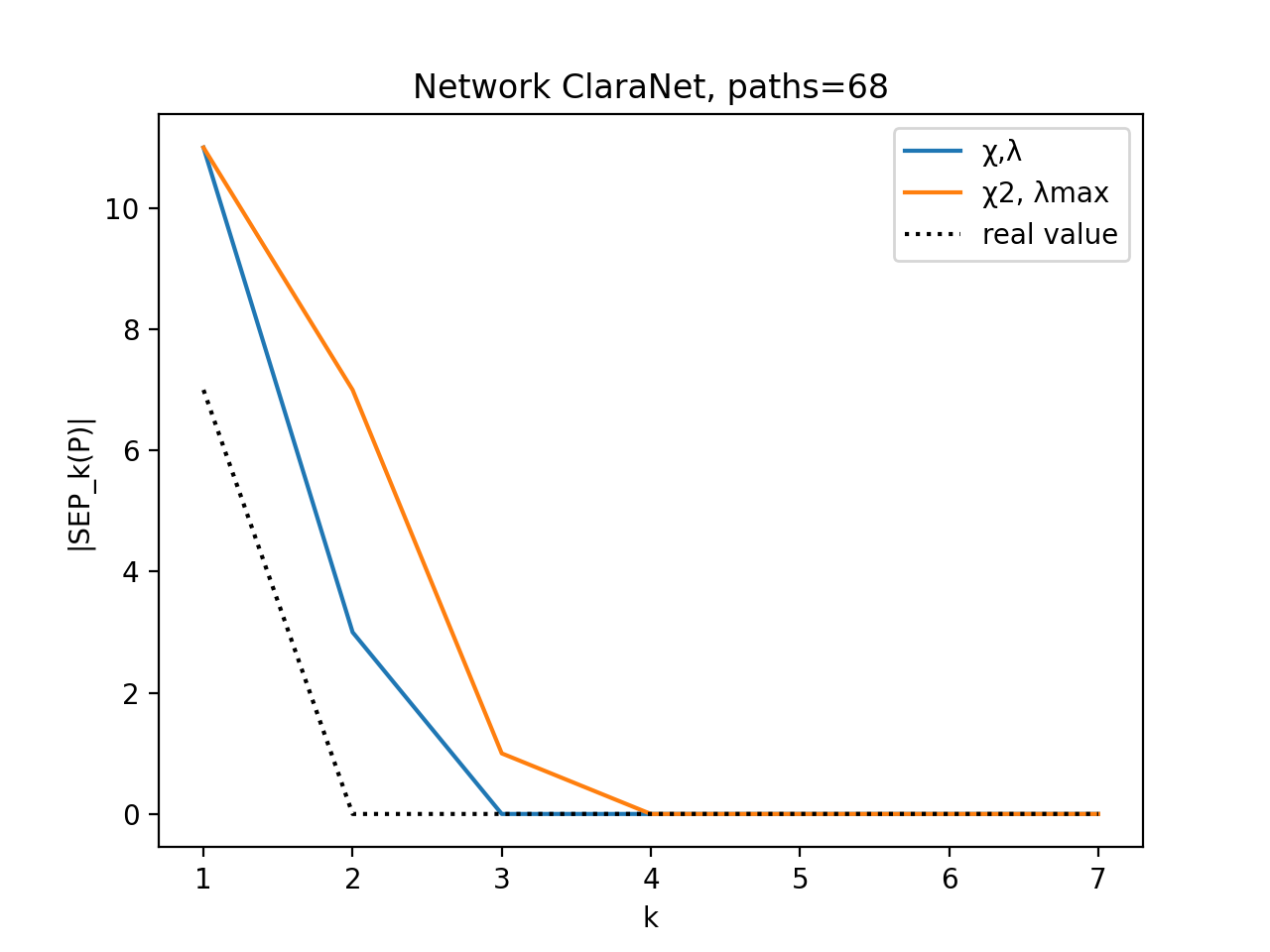}
  \end{subfigure}
  \begin{subfigure}{.3\textwidth}
  \centering
  \includegraphics[width=0.70\textwidth]{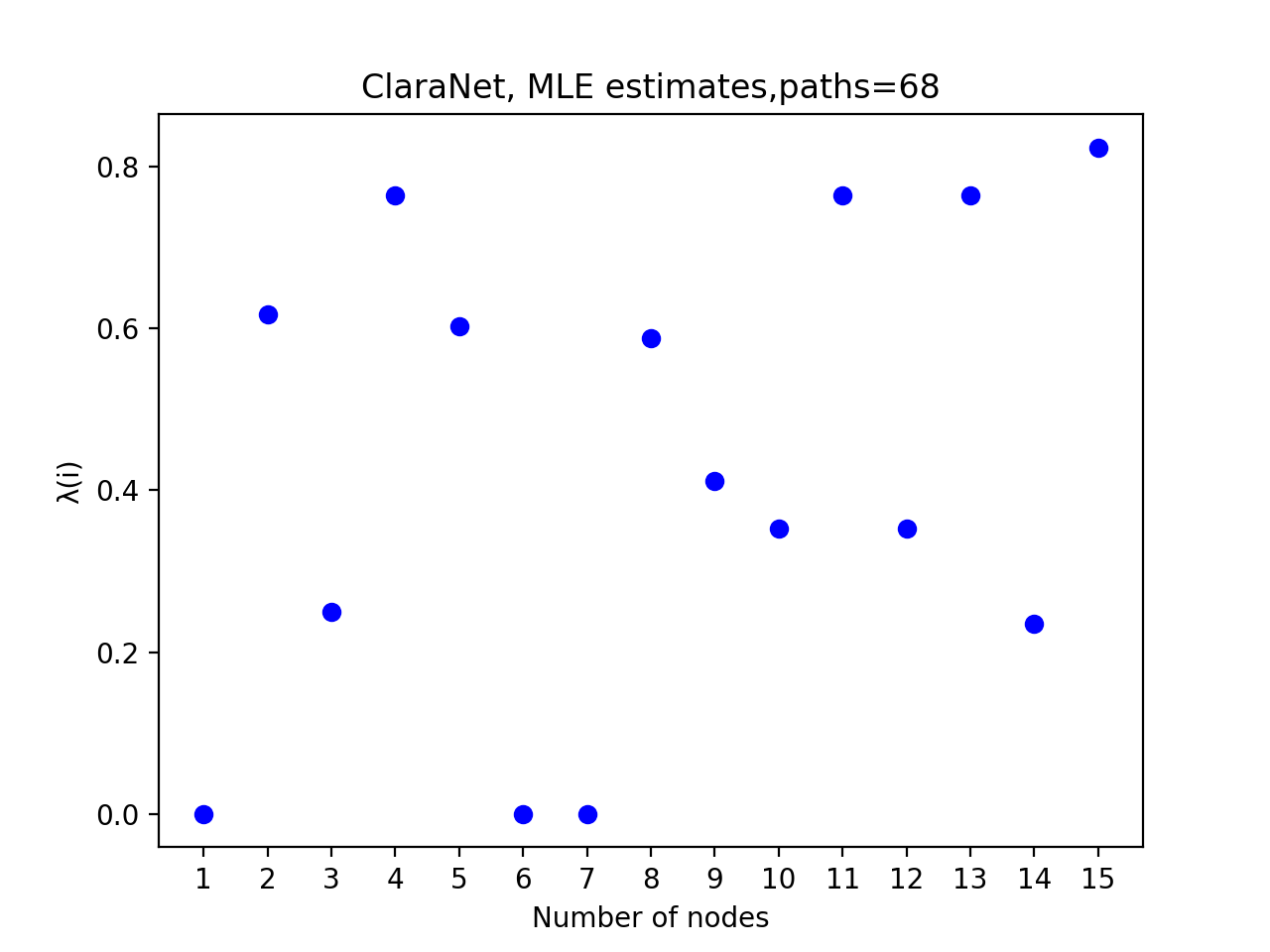}
  \end{subfigure}
   %\end{figure} 
\centering
\begin{subfigure}{.3\textwidth}
  \centering
  \includegraphics[width=0.70\textwidth]{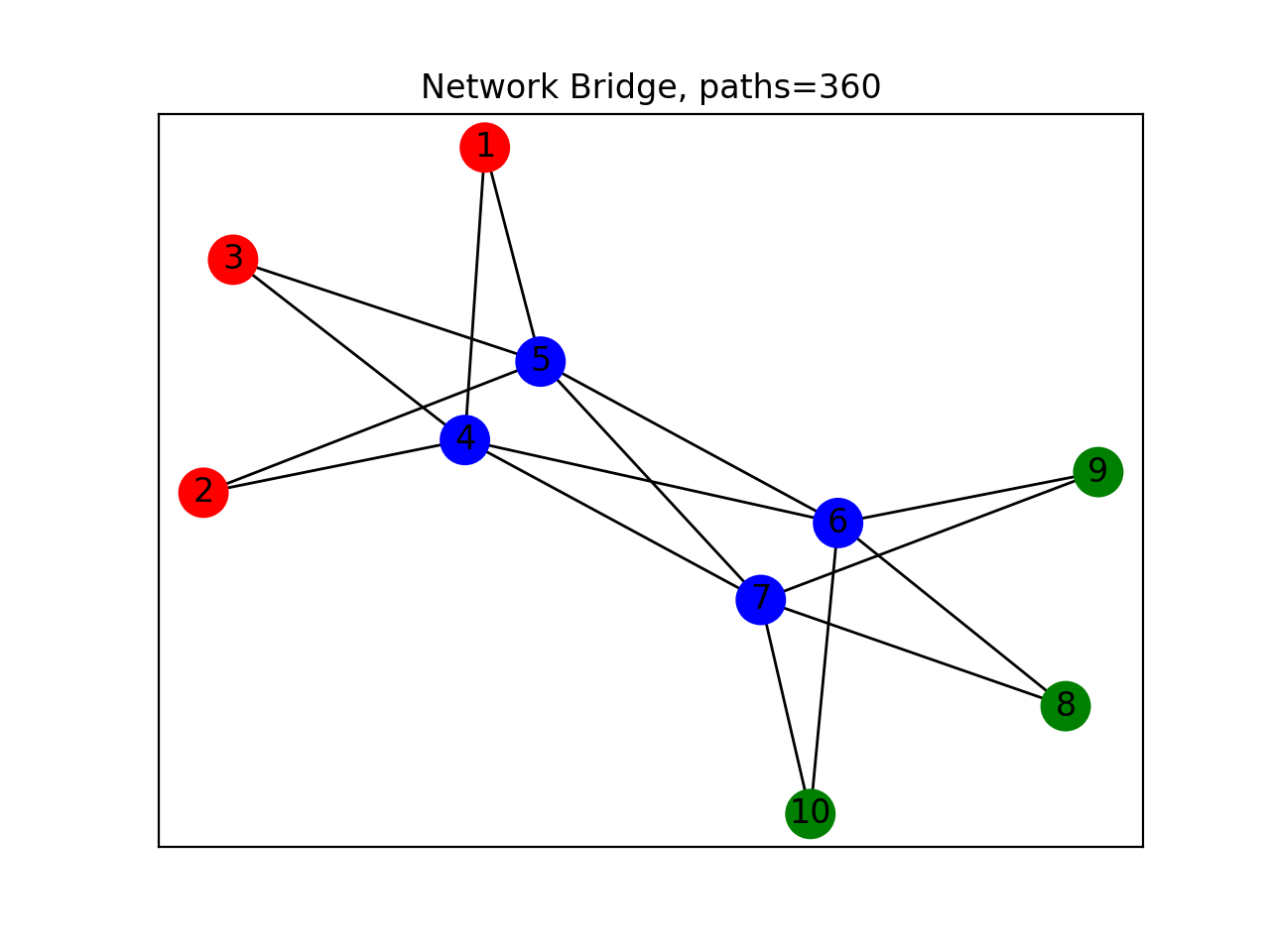}
  \end{subfigure}
  \begin{subfigure}{.3\textwidth}
  \centering
 	 \includegraphics[width=0.70\textwidth]{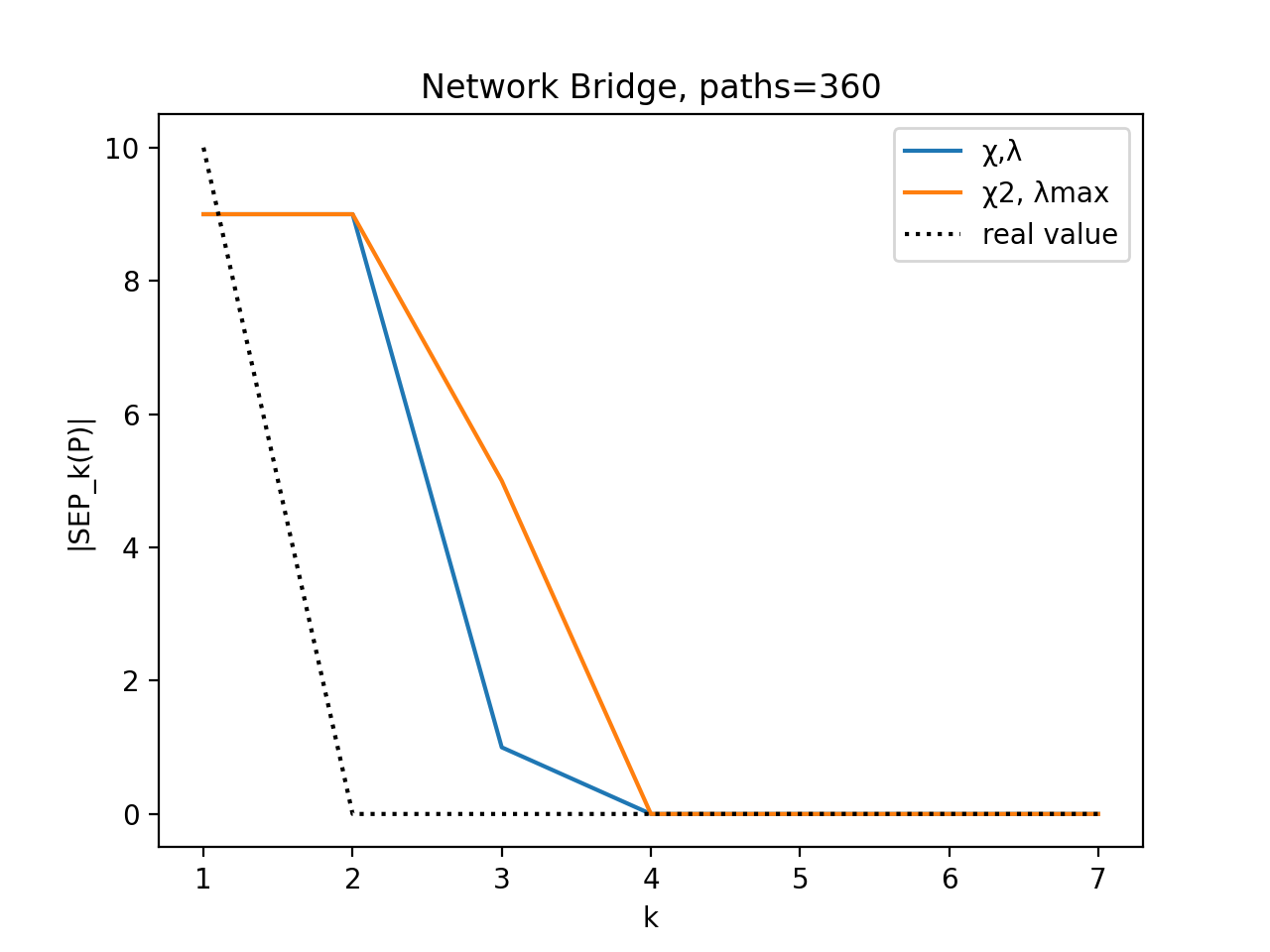}
  \end{subfigure}
  \begin{subfigure}{.3\textwidth}
  \centering
  \includegraphics[width=0.70\textwidth]{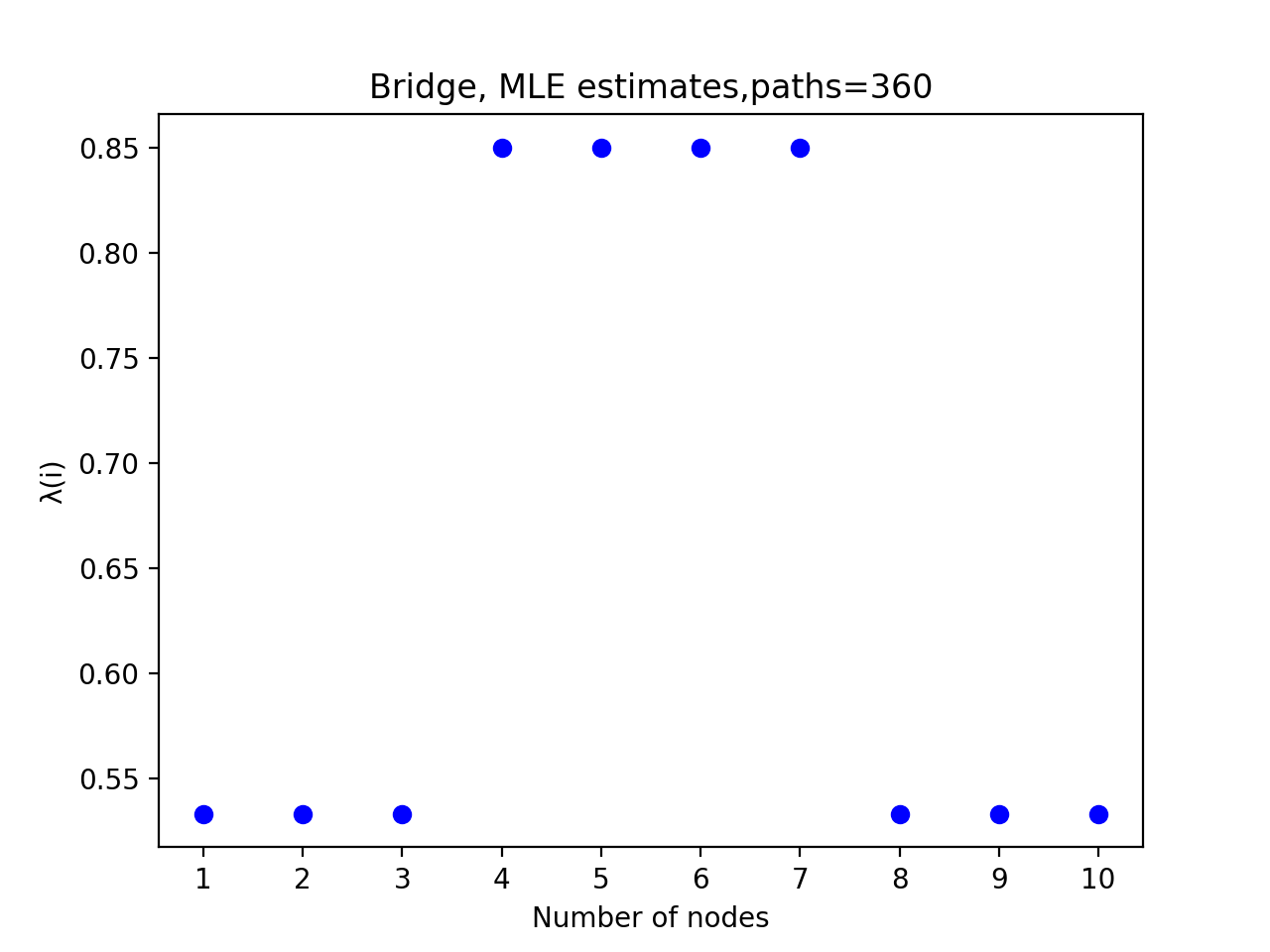}
  \end{subfigure}
  \caption{The network ClaraNet and a Bridge Network}  
     \label{fig2}

 \end{figure*}

\subsection{Experiments}
Let $\nu_{n,m,\vec \lambda}(u)=\Pr[u \in \SEP_k(\pp)]$. 
Assume to have a real set of $M$ paths $\hat \pp$ over $N$ nodes. From $\hat \pp$ we extract  the $\hat \lambda_i$ for all $i \in [N]$
and we then estimate $|\SEP_k(\hat \pp)|$ as  $\chi(\hat \pp,k,\hat \lambda)=\displaystyle \sum_{u \in [N]} \nu_{N,M,\hat \lambda}(u)$, using the closed formula in 
Theorem \ref{thm:mainprob}.   

In Figure \ref{fig1}  and \ref{fig2} we consider two graphs from the Internet topology Zoo (ClaraNet and BTEurope) and 
 we consider set of  measurement  paths $\pp$   obtained from  these networks by taking all the different  paths starting in 
 source and ending in a target node (green nodes are source and red nodes are target).  In the second table in each Figure ee compare the real values of $|SEP_k(\pp)|$ with the values of  
 $\chi(\hat \pp,k,\hat \lambda)$ for all these paths obtaining results very  tight to the real values. 
Notice that  to compute $\nu_{N,M,\hat \lambda}(u)$ we need to compute $\Pr[(u,W) \GOOD]$  for all $W \in S(k)$\footnote{This is because we need to use
 the $\lambda_w$ for all $w \in W$.}. We consider another estimates  of $|\SEP_k(\hat \pp)|$ obtained from $\chi(\hat \pp,k,\vec {\hat \lambda})$
 by having only one value for all the $\lambda_i$'s. We consider the significant case  $\chi_2(\hat \pp,k, {\hat \lambda_{\max}})$, where $\hat \lambda_{\max} = \max_i \lambda_i$\footnote{It is easy to see that $\chi(\hat \pp,k,\vec {\hat \lambda}) \geq \chi_2(\hat \pp,k, {\hat \lambda_{\max}})$.}.  Notice that in these cases we do not have to 
 have available all the $\lambda_w$ for all $W\subseteq [n]$ of size at most $k$ and the computation can be made much less expensive   since  $\chi_2(\hat \pp,k, {\hat \lambda_{\max}})=\prod_{j \in [k]} (1-(1-\lambda_{u}(1-\lambda_{\max})^j)^m)^{{n-1 \choose j}}$ and we can use methods to approximate ${n-1 \choose j}$.  The estimates $\chi_2$ is already very good in all these examples. 
 
In each table we also scatter the estimates $\hat \lambda_i$ coming from the MLE method.

\section{Complexity of $k$-identifiability and the minimum hitting set}
\label{sec:hypergraph}
Consider the optimization problem {\em Minimum Hitting Set}, $\MHS$, that  given an hypergraph 
(a set-system) ${\cal H}=(V,E)$, where $E \subseteq 2^{V}$, asks  to find the smallest
$V'\subseteq V$ such that for all $e \in E$, $V' \cap e \not =\emptyset$. $\MHS$ is a notorious 
$\NP$-complete problem \cite{DBLP:journals/jcss/AusielloDP80,Garey:2000} extending vertex cover.

We  show how to use $\MHS$ to find the minimal $k$ such that $u$ is not $k$-$\ID$ in $\pp$, i.e. 
there exits a set  of nodes $W$ of size $|W|\leq k$ , such that $\pp(u)\subseteq \pp(W)$.  

\begin{theorem}
\label{thm:MHS}
Assume $\MHS$ is solvable in polynomial time, then, deciding  whether $u$ is not $k$-$\SEP$ in $\pp$ is solvable in polynomial time. 
\end{theorem}
\begin{proof}
Consider the subset $T(u)$  of $[n]$ of those nodes touching at least a path in $\pp(u)$.
Let $Y$ be the vector of dimension $|\pp(u)|$ defined in the $j$-th coordinate  as follows: 

$$
Y[j] = \bigvee_{v \in T(u)} \pp[j,v] \quad \quad \quad j\in \pp(u)
$$ 

$Y$ has no $0$-coordinate. For otherwise there is a path in $\pp$ only touching $u$.
Hence  $Y$ has all $1$-coordinates. We consider the set-system ${\cal H}$ obtained from $\pp$  by restricting 
the columns to $T(u)$ and the rows to $\pp(u)$. 
Let  $W$ be the smallest subset of $T(u)$ provided by $\MHS$ and covering all $\pp(u)$.
Hence $u$ is not $|W|$-$\SEP$, since $\pp(u) \subseteq \pp(W)$.

The optimality of the bound is an immediate consequence of the optimality of $\MHS$. 
There is no  subset  $Z$ of $[n]$ smaller than $U$ such that $\pp(u) \subseteq \pp(Z)$, since of course
 $Z \subseteq T(u)$ and, by optimality of $\MHS$, $Z$ it cannot be smaller than $U$. 
 \end{proof}
 
The problem of finding a {\em minimal transversal} in an hypergraph is a simplification of $\MHS$ (see below) which can be decided efficiently. 
Our reduction hence suggests to implement an algorithm on concrete example of paths where we find the smallest transversal instead of the minimal hitting set.

Let us recall the following definitions from hypergraph transversal problem 
\cite{DBLP:journals/siamcomp/EiterG95}.

\begin{definition}
Let ${\cal H}=(V,E)$ be an hypergraph. A set $T \subseteq V$ is called a {\em transversal} of $H$ if it meets all the edges of $H$, i.e. if $\forall e \in E : T \cap e \not =\emptyset$. A transversal $T$ is called {\em minimal} if  no proper 
subset $T'$ of $T$ is a transversal.
\end{definition}

It is possible to find in time $O(|V||E|)$  a minimal transversal of ${\cal H}$ by the following  algorithm (see also \cite{DBLP:journals/siamcomp/EiterG95}). 
If $E=\emptyset$, then every subset of $V$ is a transversal of ${\cal H}$, hence the minimal one is 
$\emptyset$.  If $E \not =\emptyset$, let $V=\{v_1,\ldots v_n\}$. Then define:

$$
\begin{array}{l}
V_0=V  \\
V_{i+1} =\left \{ 
                \begin{array}{ll}
                       V_i & V_i \setminus \{v_i\} \mbox{ is not a transversal of ${\cal H}$} \\
                      V_i \setminus \{v_i\} & V_i \setminus \{v_i\} \mbox{ is a transversal of ${\cal H}$}
		  \end{array} \right.
\end{array} 
$$
Hence $V_n$ is a minimal transversal of ${\cal H}$. Notice however that $V_n$ it is not necessarily the 
smallest (by cardinality) transversal of ${\cal H}$. In fact this last problem is the $\MHS$ 
problem which is $\NP$-hard.

Let us call ${\tt HT}$ be a procedure that implements the previous algorithm on a  given ${\cal H}(V,E)$,  and given an order on $V$, and outputs a  minimal transversal of ${\cal H}$.

The proof of Theorem \ref{thm:MHS} suggests an algorithm to compute an upper bound on the $k$-separablity of a node $u$ in $\pp$, where instead of computing the minimum hitting set we compute a minimal transversal using 
${\tt HT}$ on any order of the variables.    

\begin{algorithm}[]
\SetAlgoLined
\KwData{$\pp,u$}
\KwResult{$(W,s)$ s.t. $\pp(u) \subseteq \pp(W)$ and  $|W|=s$ }
 \caption{Algorithm {\tt Simple}-$\SEP$}
 
$W={\tt HT}([n],\pp(u))$\;
\Return $(W,|W|)$;
\end{algorithm}

However we can think of a slightly different  algorithm which  is computing 
{\tt HT} not only once on $([n],\pp(u))$ but several times on a sequence 
of hypergraphs of decreasing complexity.

Consider the  following sets: for all 
$i \in \pp(u)$, let 
$Z(v) = \{ i \in [n] \;|\; \pp[i,v]=0\}$, for all $v \in [n]$ and $V^i =\{ v \in [n] \;|\; |\pp(u)  \cap Z(v)|=i\}$. 
Let $I=\{i_1,\ldots i_N\} \subseteq \pp(u)$ be the set of indices of the $V_{i_j} \not =\emptyset$.
We say that $V_{i_N},\ldots, V_{i_1}$ is a $0$-decreasing sequence since, by definition,
$Z(v) > Z(w)$ whenever $v \in V_i$, $w \in V_j$ and $i<j$.

\begin{algorithm}[]
\SetAlgoLined
\KwData{$\pp,u$}
 \caption{{\tt Decr}-$\SEP$}
\KwResult{$(W,s)$ s.t. $\pp(u) \subseteq \pp(W)$ and  $|W|=s$ }
 %initialization\;

Compute all $V_i$'s \;
Compute $I$ \;
$\pp_0[u]=\pp[u]$\;
 \For{$l=0 \ldots N$}{
 $k=N-l$\;
 \For{$j \in \pp(u)$}{
 \eIf{$j \in \pp_l[u]$}
  {$\vec Y_{i_k}[j]= \bigvee_{v \in V_{i_k}} \!\!\!\pp_l[j,u]$} 
 {$\vec Y_{i_k}[j]=0$ \label{line:vecY}\;}}
 $\hat V_{i_k} ={\tt HT}(V_{i_k},\pp_l[u])$\label{line:ht1}\;
 $Z_{i_k}$ = $0$-coordinates of $Y_{i_k}$\;
$\pp_{l+1}(u)=\pp_l(u) \cap Z_{i_k}$\label{line:newP}\;
 }
${\cal Y} = {\tt HT}(I, \bigcup_{i \in I} \vec Y_i)$\label{line:ht2}\; 
$W=\bigcup_{i \in {\cal Y}} \hat V_i$\;
\Return $(W,|W|)$\;
\end{algorithm}

The algorithm starts by computing the set $V_i$ and the set of indices $I$ of such sets which are not empty.
The main observations on the algorithm are the following:
\begin{itemize}
\item that $V_{i_N},\ldots, V_{i_1}$ is a $0$-decreasing sequence. At each step we try to cover only the paths in $\pp(u)$ not already covered before. This is the reason why in line \ref{line:newP} we restrict only to $0$-coordinates in $Z_{i_k}$. The vectors $Y_i$ are also defined accordingly. Only the coordinates in $\pp_l(u)$ are important since the rest are already covered by some previous $V_i$. That is the reason why in line \ref{line:vecY} we define to be $0$ the $Y$ vector in all the coordinates not  in  $\pp_l(u)$.
\item Another observation is that at each step $l$ we want to save the minimal set of nodes $\hat V_{i_{N-l}}$ sufficient to cover all the $1$'s in $\pp_l[u]$.  This is the meaning of the call to {\tt HT} in line \ref{line:ht1}.

\item finally, when we have done with  analizying all the family of the sets $V_i$'s,  $\pp(u)$ is covered by the union of the $Y$ vectors (this is by an argument similar to that fof Theorem \ref{thm:MHS}). But it is sufficient  to have the minimal subset of this family for covering all $\pp(u)$. To this end we perform a  final call to {\tt HT} on input the set-system, 
$(I, \bigcup_{i \in I} \vec Y_i)$ in \ref{line:ht2}. 
\end{itemize}

\subsection{$\NP$-Completeness}
Consider the following optimization problem MIN-NOT-$\SEP$ ($\MNS$).

\noindent{\em Input:} A Boolean $m\times n$ matrix $\pp$, an element $u \in[n]$;

\noindent{\em Output:} $k$ such that $u$ is not $k$-$\SEP$ and $u$ is $k'$-$\SEP$ for all $k'<k$.

\begin{theorem} $\MNS$ is $\NP$-complete
\end{theorem}
\begin{proof}

To see  that $\MNS$ is in $\NP$ we can use the reduction in Theorem \ref{thm:MHS} which is in fact proving that 
$\MNS \leq_p \MHS$. Since $\MHS \in \NP$ \cite{DBLP:journals/jcss/AusielloDP80}, then $\MNS \in \NP$.

To prove the $\NP$-hardness of $\MNS$ we show the opposite reduction, i.e. that $\MHS\leq_p \MNS$. Hence the result follows by the 
$\NP$-hardenss of $\MHS$ \cite{DBLP:journals/jcss/AusielloDP80}.
Let ${\cal H}=(V,E)$ be an instance of $\MHS$. We define an instance of $\MNS$ as follows:
\begin{itemize}
\item The set of nodes of $\pp$ is $V \cup \{u\}$. 
\item The set of paths of $\pp$ is $E$;
\item $\pp(u)=E$;
\end{itemize}
Since a  minimal hitting set $W$ is touching all edges in $E$, that means that
$\pp(W) = E =\pp(u)$. Hence $u$ is not $|W|$-$\SEP$. Moreover since it is minimal, then
for any subset $W'$of $[n]$ of size smaller that $|W|$ there is an edge $e \in E$ not in $W'$.
That means that $e\in \pp(v)\setminus \pp(W)$, that is $u$ is $k'$-$\SEP$ in $\pp$ for any $k'< |W|$. 

On the opposite direction,  assume that $W \subseteq [n]-\{u\}$ is witnessing that $u$ is not $|W|$-$\SEP$ but
is $k'$-$\SEP$ for any $k'<|W|$, then $W$ is clearly a minimal hitting set in ${\cal H}$.

\section{Localizing failure nodes in real networks}
\label{sec:app}
In this section we study some heuristics to compute as more precisely as possible the number of $k$-identifiable nodes in set of measurements paths defined on concrete networks, that is the the set of all paths from between 
 monitor nodes.
According to Section \ref{sec:ref}  we study upper bounds on the number of $k$-distinguishable nodes.  
%We assume the  network given as a set of paths. 
%If the network is given as a set of paths over $n$ nodes then we immediately have the associated 
%path matrix $\pp$. If the network is given as a directed graph $G=([n],E)$ over $n$ nodes, then we let $\pp$ to be defined from the set 
%of all (non-loop) paths in $G$.  
%We remark that, differently that in previous sections where we analyze combinatorial properties of $\pp$, 
%here we assume to know (as a side information to $\pp$) the order of the nodes in each path in $\pp$. As before $n$ is the number of nodes and $m$ that of paths. 

To upper bound $|\DIS_k|$ we lower bound the number of node which are {\em not} distinguishable in $\pp$.  In fact we  will localize specific sets of nodes which we can guarantee to be not $k$-distinguishable.

Let $\pp$ be given and let $u \in V$. We let $\mathbb W_k(u)$ be a subset of ${V\setminus\{u\} \choose \leq k}$.
This should be meant as (a method to generate) a collection of subset of at most $k$ nodes in $V-\{u\}$ as function of the node $u$. An example can be:  the subsets of $[n]$ made by  at most $k$ nodes which are at distance at most $d$ from $u$.
For any $v \in \mathbb W(u)$, let ${\cal P}(u,v) \subseteq \pp(u) \cap \pp(v)$. 
This should be meant as (a method to generate) a subset of all paths touching both nodes $u$ and $v$. 
 
\begin{definition}
\label{def:keq}
Let $\mathbb W$ and ${\cal P}$ be given for $\pp$. We say that $u \in [n]$ and $W \in \mathbb W_k(u)$ are {\em $k$-equal  modulo ${\cal P}$ in $\pp$} if
\begin{enumerate}
\item $\exists w,w' \in W$ such that $\pp(u)\setminus  {\cal P}(u,w)\subseteq \pp(w')$, and 
\item $\forall w \in W$, $\pp(w) \setminus {\cal P}(u,w)\subseteq \pp(u)$.
\end{enumerate} 
\end{definition}

Let 
$$
\begin{array}{l}
E_{V,k}[\mathbb W,{\cal P}] := \{u \in V : \mbox{there is a $W \in \mathbb W(u)$} \\
  \quad \quad\quad \quad\mbox{s.t. $u$ and $W$ are $k$-equal modulo ${\cal P}$\}}
\end{array}
$$
\begin{lemma}
\label{lem:eqk} 
For all $V\subseteq [n]$, $E_{V,k}[\mathbb W,{\cal P}] \subseteq \overline{\DIS_k(\pp)}$.
\end{lemma}
\begin{proof}
Let $u \in E_{V,k}[\mathbb W,{\cal P}]$ we have to find a $W \in {V \choose \leq k}$ with $u \not \in W$  such that $\pp(u)=\pp(W)$.
Fix as $W$ the one in $\mathbb W(u)$ given by the the definition of $E_{V,k}[\mathbb W,{\cal P}]$. 
We first argue that $\pp(u) \subseteq \pp(W)$.  By Definition \ref{def:keq} case (1) we know that there are $w,w'\in W$ such that
$\pp(u) -{\cal P}(u,w) \subseteq \pp(w')$. Consider a  $p\in \pp(u)$. If $p\in{\cal P}(u,w)$, then $p\in \pp(w)$ and hence $p\in \pp(W)$. If $p\not \in 
{\cal P}(u,w)$, then $p\in \pp(u) \setminus {\cal P}(u,w)$ and then by Definition \ref{def:keq} case (1)  $p \in \pp(w')$ and hence in $\pp(W)$.  

\smallskip

Let $q\in \pp(W)$, then $q \in \pp(w)$ for some $w \in W$. If $q \in  {\cal P}(u,w)$, then $q\in \pp(u)$. If $q \not \in  {\cal P}(u,w)$, then 
$q \in \pp(w) \setminus {\cal P}(u,w)$ and then, by Definition \ref{def:keq} case (1), $q \in \pp(u)$. 
\end{proof}

By Lemma \ref{lem:eqk} nodes in $E_{[n],k}[\mathbb W,{\cal P}]$ are not $k$-distinguishable and, for the anti-monotonicity, are not 
$(k+1)$-,$(k+2)$-,$\ldots$, $n$-distinguishable.

We now study how to upper bound the number of $k$-distinguishable nodes in $\pp$ given a specific definition of $\mathbb W$ and $ {\cal P}$.   
Consider  the following family of vertices in $[n]$
:$$
\left \{ 
\begin{array}{ll}
V_1=[n]  &  \\
%\displaystyle 
V_k = [n]- \bigcup_{j<k} E_{V_j,j}[\mathbb W,{\cal P}] & \quad k>1\end{array}
\right.
$$

\begin{definition} 
\label{def:tau}
 Let $k \leq n$.
$\tau_k :=|E_{V_k,k}[\mathbb W,{\cal P}]|$.
\end{definition}

\begin{theorem}
$ |\DIS_k(\pp)| \leq n - \sum_{j=1}^k \tau_j$.
\end{theorem}
\begin{proof}
We abbreviate $E_{V_j,j}[\mathbb W,{\cal P}]$ with $E_{V_j,j}$.
First we claim that  $\left | \bigcup_{j\leq k} E_{V_j,j}\right|  \leq \sum_{j=1}^k \tau_j$.
This is because  for all $k \leq n$, if $u \in E_{V_k,k}$, then   $u \not \in \bigcup_{j\leq k-1} E_{V_j,j}$, 
by definition of $E_{V_k,k}$.
 
 \smallskip
 
Further we claim that 
$$E_{V_k,k} \subseteq \overline{\DIS_k(\pp)} \setminus \bigcup_{j\leq k-1} E_{V_j,j}.$$
Indeed by Lemma  \ref{lem:eqk} $E_{V_k,k} \subseteq \overline{\DIS_k(\pp)}$ 
and again by definition of $E_{V_k,k}$, if 
$u \in E_{V_k,k}$, then $u \not \in \bigcup_{j\leq k-1} E_{V_j,j}$ . 
Therefore: 
$$\overline{\DIS_k(\pp)} \geq  |E_{V_k,k}| +\left | \bigcup_{j\leq k-1} E_{V_j,j}\right |.$$

By definition of $\tau_k$ it follows that $|\overline{\DIS_k(\pp)}| \geq \tau_k +  \sum_{j=1}^{k-1} \tau_j$, and hence that
$|\DIS_k(\pp)| \leq n - \sum_{j=1}^{k} \tau_j$.

\end{proof}

Notice that the proof of the theorem is constructive and is counting well-defined nodes in the network, so that
nodes can also precisely be localized.  

\subsection{Examples of applications}
We show how to use previous results to localize and upper bound the number of $k$-identifiable nodes on real 
set of measurements paths. The estimate will depend on 
what set $\mathbb W(u)$ we consider for any node $u$ and on what set of paths ${\cal P}$  we are going to test
path not distinguishability.  However once we have fixed $\mathbb W$ and ${\cal P}$ the algorithm we run is always the same and reflects the discussion in the previous  subsection (See Algorithm {\tt lb-$\DIS_k$}).

%Notice that $\mathbb W$ and $E_{V,i}[\mathbb W,{\cal P}]$ are given as  binary vectors of dimension $[n]$ and 
%${\cal P}(u,w)$ as a binary vector of dimension $[m]$

\begin{algorithm}[]
\SetAlgoLined
\KwData{$\pp$}
 \caption{{\tt lb-$\DIS_k$}: Counting $k$-$\SEP$ nodes}
\KwResult{number of $k$-$\SEP$ nodes }
 %initialization\;

 \For{$u \in [n]$}{
 Compute $\mathbb W(u)$\;
 \For{$w \in \mathbb W(u)$}{
  Compute ${\cal P}(u,w)$}
 }
 $V=[n], i=1, \tau=0$\;
  \While{$i\leq k$}{
  Compute $E_{V,i}[\mathbb W,{\cal P}]$\;
  $\tau = \tau+ |E_{V,i}[\mathbb W,{\cal P}]|$\;
  $V= V-E_{V,i}[\mathbb W,{\cal P}]$\;
  $i=i+1$;\
  }
\Return $n-\tau$\;
\end{algorithm}

Our method can be applied on a network given as a graph once we have decided the set of measurements 
paths. Every possible way of choosing $\mathbb W,$ and ${\cal P}$ is giving a way to count nodes which are no 
distinguishable. We caton therefore thinking of applying the method restricting for each node $u$ 
the nodes we are checking  be not distiguinshable and the effective paths we are going to to consider. 
We consider here three potential examples. We will add details on experiments and we wilkla dd other cases 
in the final version of the paper.

\subsubsection{Neighbours}
For any given $u \in[n]$, let $\mathbb W_k(u) = {N(u) \choose \leq k}$, where $N(u)$ are the  neighbours of $u$ and consider 
for all  $v \in N(u)$, the ${\cal P}^N(u)$ of the paths touching both $u$ and its neighbours $v$

\subsubsection{Nodes at a fixed distance $d$}
For any given $u \in[n]$, let $N_d(u)=\{v \in V: d(u,v)=d, d \geq 1\}$
and $\mathbb W_k(u) =  {N_d(u) \choose \leq k}$. For  $v \in N_d(u)$, the ${\cal P}^d$ of the paths touching both $u$ and  $v$.

\subsubsection{Shortest paths}
In this case we consider as set  $\mathbb W_k(u) =   {V-\{u\} \choose \leq k}$, and for all $v \in  V \setminus \{u\}$, ${\cal P}$ is the set of shortest paths from $u$ to $v$.

\section*{Acknowledgements} The authors would like to thank Navid Talebanfard to point them to  the paper
\cite{ST20}.
\end{proof}

\bibliographystyle{acm}
\bibliography{biblioBNTIEEE.bib}

\begin{thebibliography}{10}

\bibitem{DBLP:journals/jcss/AusielloDP80}
{\sc Ausiello, G., D'Atri, A., and Protasi, M.}
\newblock Structure preserving reductions among convex optimization problems.
\newblock {\em J. Comput. Syst. Sci. 21}, 1 (1980), 136--153.

\bibitem{DBLP:journals/ton/BartoliniHAMTK20}
{\sc Bartolini, N., He, T., Arrigoni, V., Massini, A., Trombetti, F., and
  Khamfroush, H.}
\newblock On fundamental bounds on failure identifiability by boolean network
  tomography.
\newblock {\em {IEEE/ACM} Trans. Netw. 28}, 2 (2020), 588--601.

\bibitem{DBLP:journals/tit/Duffield06}
{\sc Duffield, N.~G.}
\newblock Network tomography of binary network performance characteristics.
\newblock {\em {IEEE} Trans. Information Theory 52}, 12 (2006), 5373--5388.

\bibitem{DBLP:journals/siamcomp/EiterG95}
{\sc Eiter, T., and Gottlob, G.}
\newblock Identifying the minimal transversals of a hypergraph and related
  problems.
\newblock {\em {SIAM} J. Comput. 24}, 6 (1995), 1278--1304.

\bibitem{FF86}
{\sc Frankl, P., and F\"uredi, Z.}
\newblock Union-free families of sets and equations over field.
\newblock {\em Journal of Number Theory 23\/} (1986), 210--218.

\bibitem{DBLP:conf/icdcs/GalesiR18}
{\sc Galesi, N., and Ranjbar, F.}
\newblock Tight bounds for maximal identifiability of failure nodes in boolean
  network tomography.
\newblock In {\em 38th {IEEE} International Conference on Distributed Computing
  Systems, {ICDCS} 2018, Vienna, Austria, July 2-6, 2018\/} (2018), {IEEE}
  Computer Society, pp.~212--222.

\bibitem{GR20}
{\sc Galesi, N., and Ranjbar, F.}
\newblock Tight bounds to localize failure nodes on trees, grids and through
  embeddings under boolean network tomography.
\newblock {\em Submitted\/} (2020).

\bibitem{DBLP:conf/algosensors/GalesiRZ19}
{\sc Galesi, N., Ranjbar, F., and Zito, M.}
\newblock Vertex-connectivity for node failure identification in boolean
  network tomography.
\newblock In {\em Algorithms for Sensor Systems - 15th International Symposium
  on Algorithms and Experiments for Wireless Sensor Networks, {ALGOSENSORS}
  2019, Munich, Germany, September 12-13, 2019, Revised Selected Papers\/}
  (2019), F.~Dressler and C.~Scheideler, Eds., vol.~11931 of {\em Lecture Notes
  in Computer Science}, Springer, pp.~79--95.

\bibitem{Garey:2000}
{\sc Garey, M.~R., and Johnson, D.~S.}
\newblock {\em Computers and Intractability, A Guide to the Theory of
  NP-Completeness}, 22~ed.
\newblock W. H. Freeman and Company, New York, 2000.

\bibitem{DBLP:journals/ton/MaHLST14}
{\sc Ma, L., He, T., Leung, K.~K., Swami, A., and Towsley, D.}
\newblock Inferring link metrics from end-to-end path measurements:
  Identifiability and monitor placement.
\newblock {\em {IEEE/ACM} Trans. Netw. 22}, 4 (2014), 1351--1368.

\bibitem{DBLP:journals/pe/MaHSTL15}
{\sc Ma, L., He, T., Swami, A., Towsley, D., and Leung, K.~K.}
\newblock On optimal monitor placement for localizing node failures via network
  tomography.
\newblock {\em Perform. Eval. 91\/} (2015), 16--37.

\bibitem{DBLP:journals/ton/MaHSTL17}
{\sc Ma, L., He, T., Swami, A., Towsley, D., and Leung, K.~K.}
\newblock Network capability in localizing node failures via end-to-end path
  measurements.
\newblock {\em {IEEE/ACM} Trans. Netw. 25}, 1 (2017), 434--450.

\bibitem{DBLP:conf/imc/MaHSTLL14}
{\sc Ma, L., He, T., Swami, A., Towsley, D., Leung, K.~K., and Lowe, J.}
\newblock Node failure localization via network tomography.
\newblock In {\em Proceedings of the 2014 Internet Measurement Conference,
  {IMC} 2014, Vancouver, BC, Canada, November 5-7, 2014\/} (2014),
  C.~Williamson, A.~Akella, and N.~Taft, Eds., {ACM}, pp.~195--208.

\bibitem{DBLP:conf/infocom/RenD16}
{\sc Ren, W., and Dong, W.}
\newblock Robust network tomography: K-identifiability and monitor assignment.
\newblock In {\em 35th Annual {IEEE} International Conference on Computer
  Communications, {INFOCOM} 2016, San Francisco, CA, USA, April 10-14, 2016\/}
  (2016), {IEEE}, pp.~1--9.

\bibitem{ST20}
{\sc Shangguan, C., and Tamo, I.}
\newblock New tur\'an exponents for two extremal hypergraph problems.
\newblock {\em SIAM Journal on Discrete Mathematics\/} (2020).

\end{thebibliography}

\end{document}